\documentclass[11pt]{article}

\usepackage[letterpaper,margin=1in]{geometry}
\usepackage{hyperref}
\usepackage{todonotes}
\usepackage{algorithm}
\usepackage{amsthm}
\usepackage{amsmath, amssymb}
\usepackage[capitalize]{cleveref}
\usepackage{comment}
\usepackage{complexity}
\usepackage{algorithmic}

\newcommand{\ED}{\textnormal{ED}}
\newcommand{\LCS}{\textnormal{LCS}}
\newcommand{\RSD}{\textnormal{RSD}}
\newcommand{\eps}{\varepsilon}
\newcommand{\Decode}{\textnormal{Dec}}

\newcommand{\IEEEPARstart}[2]{#1#2}

\newtheorem{theorem}{Theorem}[section]
\newtheorem{definition}[theorem]{Definition}
\newtheorem{lemma}[theorem]{Lemma}
\newtheorem{proposition}[theorem]{Proposition}

\newcommand{\IEEEOnly}[1]{}

\begin{document}

\date{}

\title{Synchronization Strings and Codes for \\Insertions and Deletions -- a Survey}



\author{Bernhard Haeupler\footnote{Supported in part by NSF grants CCF-1527110, CCF-1618280, CCF-1814603, CCF-1910588, NSF CAREER award CCF-1750808, a Sloan Research Fellowship, and funding from the European Research Council (ERC) under the European Union's Horizon 2020 research and innovation program (ERC grant agreement 949272).}
\\Carnegie Mellon University \& ETH Zurich\\ \texttt{haeupler@cs.cmu.edu} \and Amirbehshad Shahrasbi\footnotemark[1]\\Carnegie Mellon University\\ \texttt{shahrasbi@cs.cmu.edu}}

\maketitle
\thispagestyle{empty}

\begin{abstract}
Already in the 1960s, Levenshtein and others studied error-correcting codes that protect against synchronization errors, such as symbol insertions and deletions. However, despite significant efforts, progress on designing such codes has been lagging until recently, particularly compared to the detailed understanding of error-correcting codes for symbol substitution or erasure errors. This paper surveys the recent progress in designing efficient error-correcting codes over finite alphabets that can correct a constant fraction of worst-case insertions and deletions. 

Most state-of-the-art results for such codes rely on synchronization strings, simple yet powerful pseudo-random objects that have proven to be very effective solutions for coping with synchronization errors in various settings. This survey also includes an overview of what is known about synchronization strings and discusses communication settings related to error-correcting codes in which synchronization strings have been applied. 

\end{abstract}

\thispagestyle{empty}

\newpage
\setcounter{page}{1}

\section{Introduction}
\IEEEPARstart{F}{ollowing} the inspiring works of Shannon and Hamming
a sophisticated and extensive body of research on error-correcting codes has led to a deep and detailed theoretical understanding as well as practical implementations that have helped fuel the Digital Revolution. Error-correcting codes can be found in virtually all modern communication and computation systems. While being remarkably successful in understanding the theoretical limits and trade-offs of reliable communication under substitution errors and erasures, the coding theory literature lags significantly behind when it comes to overcoming errors that concern the timing of communications. In particular, the study of correcting synchronization errors, i.e., symbol insertions and deletions, while initially introduced by Levenshtein in the 60s, has significantly fallen behind our highly sophisticated knowledge of codes for Hamming-type errors, that are symbol substitutions and erasures.

This discrepancy has been well noted in the literature. An expert panel~\cite{golomb1963synchronization} in 1963 concluded: \emph{``There has been one glaring hole in [Shannon's] theory; viz., uncertainties in timing, which I will propose to call time noise, have not been encompassed \ldots. Our thesis here today is that the synchronization problem is not a mere engineering detail, but a fundamental communication problem as basic as detection itself!''} however as noted in a comprehensive survey~\cite{mercier2010survey} in 2010: \emph{``Unfortunately, although it has early and often been conjectured that error-correcting codes capable of correcting timing errors could improve the overall performance of communication systems, they are quite challenging to design, which partly explains why a large collection of synchronization techniques not based on coding were developed and implemented over the years.''} or as Mitzenmacher puts in his survey~\cite{mitzenmacher2009survey}: \emph{``Channels with synchronization errors, including both insertions and deletions as well as more general timing errors, are simply not adequately understood by current theory. Given the near-complete knowledge we have for channels with erasures and errors \ldots our lack of understanding about channels with synchronization errors is truly remarkable.''}

However, over the last five years, partially spurred by new emerging application areas, such as DNA-storage~\cite{organick2017scaling,blawat2016forward,goldman2013towards,church2012next,yazdi2015dna,bornholt2016dna}, significant breakthroughs in our theoretical understanding of error correction methods for insertions and deletions have been made. 

This survey focuses on error-correcting codes over finite alphabets that can correct a constant fraction of worst-case insertions and deletions and provides a complete account of the recent progress in this area. Much of this progress has been obtained through synchronization strings, recently introduced, simple yet powerful pseudo-random objects proven to be very effective solutions for coping with synchronization errors in various communication settings. This paper includes streamlined and self-contained proofs for the state-of-the-art code constructions and decoding procedures for both unique-decodable and list-decodable error-correcting codes over large constant alphabets, which are based on synchronization strings. We also provide in-depth discussions of such codes over binary and other fixed (small) alphabets. Lastly, this paper includes an overview of what is known about synchronization strings themselves and discusses other communication settings in which synchronization strings have been successfully applied.



\subsection{Synchronization Errors}
Consider a stream of symbols being transmitted through a noisy channel.
There are two basic types of errors that we will consider, Hamming-type errors and synchronization errors. 
\emph{Hamming-type errors} consist of \emph{erasures}, that is, a symbol being replaced with a special ``?'' symbol indicating the erasure, and \emph{substitutions} in which a symbol is replaced with any other symbol of the alphabet. We will measure Hamming-type errors in terms of \emph{half-errors}. The wording half-error comes from the realization that, when it comes to code distances, erasures are half as bad as symbol substitutions. An erasure is thus counted as one half-error while a symbol substitution counts as two half-errors. 
\emph{Synchronization errors} consist of \emph{deletions}, that is, a symbol being removed without replacement, and \emph{insertions}, where a new symbol is added somewhere within the stream. 

Synchronization errors are strictly more general and harsher than half-errors. In particular, any symbol substitution, worth two half-errors, can also be achieved via a deletion followed by an insertion. Any erasure can be interpreted as a deletion together with the extra information where this deletion has taken place. This shows that any error pattern generated by $k$ half-errors can also be replicated using $k$ synchronization errors, making dealing with synchronization errors at least as hard as half-errors. The real problem that synchronization errors bring about, however, is that they cause sending and receiving parties to become ``out of sync''. This easily changes how received symbols are interpreted and makes designing codes or other systems tolerant to synchronization errors an inherently difficult and significantly less well-understood problem. 

\subsection{Scope of the Survey and Related Works}
The study of coding for synchronization errors was initiated by Levenshtein~\cite{levenshtein1966binary} in 1966 when he showed that Varshamov-Tenengolts codes can correct a single insertion, deletion, or substitution error with an optimal redundancy of almost $\log n$ bits. 
Ever since, synchronization errors have been studied in various settings. In this section, we specify and categorize some of the commonly studied settings and give a detailed summary of past works within the scope of this survey.

The first important aspect is the noise model. Several works have studied coding for synchronization errors under the assumption of random errors, most notably, to study the capacity of deletion channels, which independently delete each symbol with some fixed probability. In this paper, we exclusively focus on worst-case error models in which correction has to be possible from any (adversarial) error pattern bounded only by the total number of insertions and deletions. We refer to the recent survey (in the same special issue) by Cheraghchi and Ribeiro~\cite{cheraghchi2019overview} on capacity results for synchronization channels as well as the surveys by Mitzenmacher~\cite{mitzenmacher2009survey} and Mercier~\cite{mercier2010survey}, for an extensive review of the literature on codes for random synchronization errors.

Another angle to categorize the previous work on codes for synchronization error from is the noise regime. In the same spirit as ordinary error-correcting codes, the study of families of synchronization codes has included both ones that protect against a fixed number of synchronization errors and ones that consider error count that is a fixed fraction of the block length. The inspiring work of Levenshtein~\cite{levenshtein1966binary} falls under the first category and is followed by several works designing synchronization codes correcting $k$ errors for specific values of $k$~\cite{sloane2002single,tenengolts1984nonbinary,helberg2002multiple,gabrys2018codes} or with $k$ as a general parameter~\cite{abdel2011helberg,brakensiek2017efficient}. In this work, we focus on the second category, i.e., infinite families of synchronization codes with increasing block length that are defined over a fixed alphabet size and can correct from constant-fractions of worst-case synchronization errors. 

Furthermore, we mainly focus on codes that can be efficiently constructed and decoded -- in contrast to merely existential results.
The first such code was constructed in 1990 by Schulman and Zuckerman~\cite{schulman1999asymptotically}. They provided an efficient, asymptotically good synchronization code with constant rate and constant distance. 
In the following, we will give a complete review of the previous work relevant to the scope of this paper.

\medskip

\subsubsection{Rate-Distance Trade-Off}
One of the main problems in coding theory concerns the question of what the largest achievable communication rate is while protecting from a certain fraction of (synchronization) errors. This question can be studied under the regime that assumes some fixed alphabet of size $q$, specifically binary alphabets, or an alphabet-free regime that studies the rate achievability when alphabet size can be chosen arbitrarily large but independent of the block length.

For the large alphabet setting, the Singleton bound suggests that no family of codes can correct a $\delta$ fraction of deletions, and hence, $\delta$ fraction of synchronization errors while achieving a rate strictly larger than $1-\delta$.
A series of works by Guruswami \emph{et al}.~\cite{guruswami2016efficiently,guruswami2017deletion} provides codes that achieve a rate of $\Omega((1-\delta)^5)$ and $1 - \tilde{O}(\sqrt{\delta})$ while being able to efficiently recover from a $\delta$ fraction of insertions and deletions in high-noise and high-rate regimes respectively. In this paper, we will take a deep dive into synchronization string based code constructions that provide codes that can approach the Singleton bound up to an arbitrarily small additive term over the entire distance spectrum $\delta\in(0, 1)$.

For binary alphabet codes, one can show that the optimal achievable rate to protect against a $\delta$ fraction of insertions or deletions is $1-O(\delta \log\frac{1}{\delta})$~\cite{levenshtein1966binary}. Works of Guruswami \emph{et al}.~\cite{guruswami2016efficiently,guruswami2017deletion} and Haeupler \emph{et al}.~\cite{haeupler2017synchronization2} present efficient codes with distance $\delta$ and rate $1-O\left(\sqrt{\delta}\log ^{O(1)}\frac{1}{\delta}\right)$ for sufficiently small $\delta$. Recent works by Cheng \emph{et al}.~\cite{cheng2018deterministic} and Haeupler~\cite{haeupler2018optimal} have achieved codes with rate $1-O(\delta\log^2\frac{1}{\delta})$.

\medskip
\subsubsection{List Decoding}
Like error-correcting codes, synchronization codes have been studied under the \emph{list decoding} model where, as opposed to unique decoding, the decoder is expected to produce a list of codewords containing the transmitted codeword as long as the error rate is sufficiently small.
    
Guruswami and Wang~\cite{guruswami2017deletion} have provided positive-rate binary deletion codes that can be list-decoded from close to $\frac{1}{2}$ fraction of deletions. 
Haeupler \emph{et al}.~\cite{haeupler2018synchronization4,haeupler-list-dec-capacity2020} gave upper and lower bounds on the maximum achievable rate of list-decodable insertion-deletion codes (or insdel codes for short) over any alphabet size $q$.
Recent works of Wachter-Zeh~\cite{wachter2017list} and Hayashi and Yasunaga~\cite{hayashi2018list} have studied list-decoding by providing Johnson-type bounds for synchronization codes that relate the minimum edit-distance of the code to its list decoding properties. We generally define the edit-distance between two strings as the smallest number of insertions and deletions needed to convert one to another.
The bounds presented in \cite{hayashi2018list} show that binary codes by Bukh, Guruswami, and H\aa stad~\cite{bukh2017improved} can be list-decoded from a fraction $\approx 0.707$ of insertions. Via a concatenation scheme used in \cite{guruswami2017deletion} and \cite{guruswami2016efficiently}, Hayashi and Yasunaga furthermore made these codes efficient. A recent work of Liu, Tjuawinata, and Xing~\cite{liu2019list} also derives bounds on list-decoding radius,  provides efficiently list-decodable insertion-deletion codes over small alphabets, and gives a Zyablov-type bound for synchronization codes.

\medskip
\subsubsection{Error Resilience}
As mentioned above, it is known that there exist positive-rate binary deletion codes that are list-decodable from any fraction of errors smaller than $\frac{1}{2}$. Also, there are codes that can list-decode from a fraction $\approx 0.707$ of insertions.
We will present a recent result from \cite{guruswami2019optimally} that, for any alphabet size $q$, precisely identifies the maximal rates of combinations of insertion and deletion errors from which list-decoding is possible.

A similar question can be asked for uniquely-decodable synchronization codes, i.e., what is the largest fraction of errors $\delta_0$ where there exist positive-rate synchronization codes with minimum edit-distance $\delta_0$? For binary alphabets, it is easy to see that $\delta_0 \leq \frac{1}{2}$. However, most resilient binary codes with positive rate to date are ones introduced by Bukh, Guruswami, and H\aa stad~\cite{bukh2017improved} that can correct a $\sqrt{2}-1 \approx 0.4142$ fraction of errors. Determining the optimal error resilience for uniquely-decodable synchronization codes remains an interesting open question. We refer the reader to \cite{cheraghchi2019overview} for a more comprehensive review of past works on the error resilience for synchronization codes.

\subsection{Coding with Synchronization Strings}\label{sec:statement-of-main-results}
One commonly studied approach to correct from synchronization errors is to use special symbols or sequences with specific structures as markers or delimiters to keep track of insertions and deletions and realign a transmitted word~\cite{sellers1962bit,morita1997prefix,ferreira1997insertion,gilbert1960synchronization,guibas1978maximal,van1995extended,morita1996construction,kautz1965fibonacci}. In this work, we focus on a very recent form of such technique -- indexing with synchronization strings.

Introduced in \cite{haeupler2017synchronization}, synchronization strings allow efficient synchronization of streams that are affected by insertions and deletions using an abstract indexing scheme. Essentially, synchronization strings enable compartmentalization of coding against synchronization errors into two steps of (1) realigning the received stream of symbols in a way that guarantees most symbols are in their original position and (2) coding against Hamming-type errors caused by wrong realignments. Synchronization strings have made progress on a wide variety of settings and problems. This survey focuses on code constructions that are based on synchronization strings. Most importantly, we will review the following results.

\medskip
\subsubsection{Codes Approaching the Singleton Bound}\label{sec:intro-main-result-1}
Synchronization strings enable construction of families of synchronization codes that approach an almost optimal rate-distance trade-off as suggested by the Singleton bound over constant alphabet sizes. In other words, as shown in~\cite{haeupler2017synchronization}, for any $0\leq \delta < 1$ and any $\eps > 0$, there exists a family of synchronization codes that can uniquely and efficiently correct any $\delta$ fraction of insertions and deletions and achieve a rate of $1-\delta-\eps$. Such codes exist over alphabets of size $\exp(1/\eps)$ which is shown in~\cite{haeupler2017synchronization3} to be the asymptotically optimal alphabet size for a code with such properties.

\medskip
\subsubsection{Near-Linear Time Codes}
We then present an improvement from \cite{Rubinstein18-blog} over the result just described that modifies the construction and decoding in a way that enables near-linear time decoding. Two main ingredients are used to achieve this improvement: (1) generalizations of synchronization strings and their fast construction methods introduced in~\cite{haeupler2017synchronization3}, and (2) a fast indexing scheme for edit-distance computation from~\cite{haeupler2019near}. For any $n$ and $\eps > 0$, \cite{haeupler2019near} gives string $I$ of length $n$ over an alphabet of size $|\Sigma| = O_\eps(1)$ which enables fast approximation of the edit distance in the following way: Let $S\in\Sigma'^n$ be another string of length $n$ over some other alphabet $\Sigma'$. If one concatenates $S$ and $I$, symbol-by-symbol, to obtain the string $S\times I\in\left(\Sigma\times\Sigma'\right)^n$, then edit distance from any other string $S'\in\left(\Sigma\times\Sigma'\right)^*$ to $S\times I$ can be approximated within a multiplicative factor of $1+\eps$ in near-linear time.

\medskip
\subsubsection{List Decoding for Insertions and Deletions}\label{sec:main-result-list-dec}
We then proceed to present a recent result on list-decodable synchronization codes. Using a similar synchronization string-based approach, \cite{haeupler2018synchronization4} shows that for every $0\leq \delta < 1$, every $0 \leq \gamma < \infty$ and every $\eps > 0$ there exist a family of codes with rate $1 - \delta - \eps$, over an alphabet of constant size $q = O_{\delta,\gamma,\eps}(1)$ that are list-decodable from a $\delta$-fraction of deletions and a $\gamma$-fraction of insertions. This family of codes are efficiently decodable and their decoding list size is sub-logarithmic in terms of the code's block length. We stress that the fraction of insertions can be arbitrarily large (even more than 100\%) and the rate is independent of this parameter.

\medskip
\subsubsection{Optimal Error Resilience for List Decoding}\label{sec:main-result-resilience}
Finally, we review a result by Guruswami \emph{et al}. \cite{guruswami2019optimally} that, using a code concatenation scheme for synchronization codes with codes from \cite{bukh2017improved} and \cite{haeupler2018synchronization4}, exactly identifies the maximal fraction of insertions and deletions that can be tolerated by $q$-ary list-decodable codes with non-vanishing information rate. This includes efficient binary codes that can be list-decoded from any $\delta$ fraction of deletions and $\gamma$ fraction of insertions as long as $2\delta + \gamma < 1$. One can show that list decoding is not possible for any family of codes achieving positive rates for any error fraction out of this region. Guruswami \emph{et al}.~\cite{guruswami2019optimally} have generalized this result to alphabets of size $q$ and identified the feasibility region for $(\gamma, \delta)$ as a more complex region with a piece-wise linear boundary.


\subsection{Organization of the Paper}
In \cref{sec:main-results}, we will provide proofs for claims presented in \cref{sec:statement-of-main-results} by formally introducing indexing based code constructions and giving a minimal introduction to pseudo-random strings used for indexing. 
%
%
%
In \cref{sec:sync-strings}, we discuss several pseudo-random string properties, their constructions, their repositioning algorithms and the decoding properties that they enable once used to construct codes.
We then mention applications of synchronization strings and related string properties in other communication problems such as coding for block errors and interactive communication under synchronization errors in \cref{sec:other-applications}.

\section{Code via Indexing}\label{sec:main-results}
In this section, we explain the construction of codes stated in \cref{sec:statement-of-main-results}. We start with a self-contained simplified proof of Singleton bound approaching codes presented in \cref{sec:intro-main-result-1} that encapsulates the major ideas behind synchronization string-based code constructions while avoiding unnecessary details.

\subsection{Approaching the Singleton Bound: Technical Warm-up}\label{sec:technical-warm-up}

We start by defining the notion of \emph{$\eps$-self-matching strings} that satisfy a weaker property than synchronization strings but can be used in a similar fashion to construct synchronization codes.

\begin{definition}
String $S\in\Sigma^n$ is $\eps$-self-matching if it contains no two identical non-aligned subsequences of length $n\eps$ or more, i.e., there exist no two sequences $a_1, a_2, \ldots, a_{\lfloor n\eps \rfloor}$ and $b_1, b_2, \ldots, b_{\lfloor n\eps \rfloor}$ where for all $i$s $a_i \neq b_i$ and $S[a_i] = S[b_i]$.
\end{definition}

\subsubsection{Pseudo-random Property}\label{sec:pseudo-random-property}
We first point out that random strings over an alphabet of size $\Omega(\eps^{-2})$ satisfy $\eps$-self-matching property with high probability. Note that the probability of two given non-aligned subsequences of length $n\eps$ in a random string over alphabet $\Sigma$ being identical is $\frac{1}{|\Sigma|^{n\eps}}$. Also, there are no more than ${n \choose n\eps}^2$ pairs of such subsequences. Therefore, by the union bound, the probability of such random string satisfying $\eps$-self-matching property is 
${n\choose n\eps}^2 \frac{1}{|\Sigma|^{n\eps}}\leq \left(\frac{ne}{n\eps}\right)^{2n\eps} \frac{1}{|\Sigma|^{n\eps}}=\left(\frac{e^2}{|\Sigma|\eps^2}\right)^{n\eps}$ and thus, if $|\Sigma| = \Omega(\eps^{-2})$, the random string would satisfy the $\eps$-self-matching property with high probability.

\medskip
\subsubsection{Indexing Scheme}\label{sec:indexing-scheme}
Consider a communication channel where a stream of $n$ message symbols are communicated from the sender to the receiver and assume that the communication may suffer from up to $n\delta$ adversarial insertions or deletions for some $0 \leq \delta < 1$. We introduce a simple indexing scheme that will be used to construct synchronization codes. Let $m_1, m_2, \ldots, m_n$ represent the message symbols that the sender wants to get to the receiver and $s_1, s_2, \ldots, s_n$ be some $\eps$-self-matching string that the sender and the receiver have agreed upon beforehand. To communicate its message to the receiver, we have the sender send the sequence 
$(m_1, s_1), (m_2, s_2), \ldots, (m_n, s_n)$ through the channel. We will refer to this sequence as $m$ indexed by $s$ and denote it by $m\times s$.
Note that in this setting a portion of the channel alphabet is designated to the $\eps$-self-matching string and thus, does not contain information. This portion will be used to reposition the message symbols on the receiving end of the communication as we will describe in the next section.

\medskip
\subsubsection{Repositioning (Decoding)}
We now show that, having the indexing scheme described above, the receiver can correctly identify the positions of most of the symbols it receives. Let us denote the sequence of symbols arriving at the receiving end by $(m'_1, s'_1), (m'_2, s'_2), \ldots, (m'_{n'}, s'_{n'})$. We show the following.
\begin{lemma}\label{lem:global-decoding}
There exists an algorithm for the receiving party that, having 
$(m'_1, s'_1), \ldots, (m'_{n'}, s'_{n'})$
and 
$s_1, \ldots, s_n$,
guesses the \emph{position} of all received symbols in the sent string such that positions of all but $O(n\sqrt{\eps})$ of the symbols that are not deleted in the channel are guessed correctly. This algorithm runs in $O_\eps(n^2)$ time.
\end{lemma}

Note that if no error occurs, the receiver expects the index portion of the received symbols to be similar to the $\eps$-self-matching string $s$. Having this observation, we present the decoding procedure in \cref{alg:global-decoder}.
The decoding algorithm calculates the longest common subsequence (LCS) between the synchronization string, $s$, and the index portion of the received string, $s'$. It then assigns each of the symbols from the received string that appear in the common subsequence to the position of the symbol from $s$ that corresponds to it under the common subsequence.
The algorithm repeats this procedure $1/\sqrt{\eps}$ times and after each round eliminates received symbols whose positions are guessed.

\begin{algorithm}
\caption{Insertion-Deletion Decoder}
\begin{algorithmic}[1]\label{alg:global-decoder}
\REQUIRE $s$, $(m'_1, s'_1), \cdots, (m'_{n'}, s'_{n'})$
\smallskip

\STATE $L = \left[s'_1, s'_2, \cdots, s'_{n'}\right]$
\FOR {$i = 1$ to  $n'$}
\STATE ${Position}[i] \leftarrow \texttt{Undetermined}$
\ENDFOR

\smallskip

\FOR {$i = 1$ to  $\frac{1}{\sqrt{\eps}}$}
\STATE Compute $\LCS(s, L)$
\FORALL{Corresponding $s[i]$ and $L[j]$ in $\LCS(s, L)$}\label{line:LCS-pairs}
\STATE $Position[j] \leftarrow i$
\ENDFOR
\STATE Remove all elements of $\LCS(s, L)$ from $L$
\ENDFOR

\smallskip

\ENSURE $Position$
\end{algorithmic}
\end{algorithm}

\begin{proof}[Proof of \cref{lem:global-decoding}]
Clearly, \cref{alg:global-decoder} takes quadratic time as it mainly runs $O_\eps(1)$ instances of {\LCS} computation over strings of length $O(n)$. 

To prove the correctness guarantee, we remark that there are two types of incorrect guesses for symbols that are not deleted by the adversary and bound the number of incorrect guesses of each type.
\begin{enumerate}
    \item[I)] \emph{The position of the received symbol remains \texttt{Undetermined} by the end of the algorithm:}
    Note that if by the end of the algorithm there are $k$ original symbols--i.e., symbols that are originally sent by the sender and not inserted by the adversary--that have undetermined positions, then the remainder of $L$ after $1/\sqrt{\eps}$ rounds has a common subsequence of size $k$ with $s$. This implies that, in each round of the for loop, $|\LCS(s, L)| \geq k$. Note the total size of these {\LCS}s cannot exceed the initial size of $L$ that is $n'$. Therefore,
    $k\cdot\frac{1}{\sqrt{\eps}} \leq n' \leq 2n \Rightarrow k \leq 2\sqrt{\eps}n$.
    \item[II)] \emph{The position of the received symbol is incorrectly guessed in one recurrence of the for loop:} We claim that the number of such wrong assignments in each round of the for loop is no more than $n\eps$. Let $s[i]$ and $L[j]$ be corresponding elements under $\LCS(s, L)$ in Line~\ref{line:LCS-pairs} while the received symbol that $L[j]$ identifies is the $i'$th symbol sent by the sender. This implies that $s[i] = L[j] = s[i']$. If there are more than $n\eps$ such incorrect guesses in one {\LCS} computation, we have $n\eps$ such pairs of identical symbols in $s$ that constitute a self-matching of size $n\eps$ in $s$ and violate the assumption of $s$ being an $\eps$-self-matching string. Therefore, overall there are no more than $\frac{1}{\sqrt{\eps}}\cdot n\eps = n\sqrt{\eps}$ incorrect determination of the original positions of received symbols.\qedhere
\end{enumerate}
\end{proof}

\subsubsection{Codes Approaching the Singleton Bound}
We now use the discussions on $\eps$-self-matching strings and \cref{lem:global-decoding} to construct efficient synchronization codes that can approach the Singleton bound. 

\begin{theorem}\label{thm:main-1}
For any $\eps>0$, $\delta \in (0,1)$, and sufficiently large $n$, there exists an encoding map $E: \Sigma^k \rightarrow \Sigma^n$ and a decoding map $D: \Sigma^* \rightarrow \Sigma^k$, such that, if $\ED(E(m),x) \leq \delta n$ then $D(x) = m$. Further, the rate is $\frac{k}{n} > 1 - \delta - \eps$, $|\Sigma|=\exp(1/\eps)$, and $E$ and $D$ are explicit and can be computed in linear and quadratic time in $n$. 
\end{theorem}

We use $\ED(x, y)$ to denote the edit distance between $x$ and $y$.
Note that the indexing scheme from \cref{sec:indexing-scheme} and \cref{lem:global-decoding} essentially gives a way to reduce insertions and deletions to symbol substitutions and erasures at the cost of designating a portion of the message symbols to an $\eps$-self-matching string. More precisely, with the indexing scheme from \cref{sec:indexing-scheme} in place, a receiver can use \cref{alg:global-decoder} to guess the position of the symbols it receives in the sent message and rearrange them to recover the message sent by the sender. 

Let $\tilde{m}$ denote the recovered message and $Position$ denote the output of \cref{alg:global-decoder}.
More precisely, for any $1\leq i\leq n$, the decoder sets $\tilde{m}[i] = j$ if only for one value of $j$, $Position[j] = i$. If there are zero or multiple received symbols that are guessed to be at position $i$, the decoder simply decides $\tilde{m}[i] = \texttt{?}$.

We claim that $\tilde{m}$ is different from $m$ by no more than $n(\delta + 12\sqrt{\eps})$ half-errors. Note that if an adversary applies no errors and \cref{alg:global-decoder} guesses the positions perfectly, $\tilde{m} = m$. In the following steps, we add these imperfections and see the effect in the Hamming distance between $m$ and $\tilde{m}$.
\begin{itemize}
    \item Each deleted symbol turns a detection in $\tilde{m}$ to a \texttt{?} and, therefore, adds one half-error to the Hamming distance between $m$ and $\tilde{m}$.
    \item Each inserted symbol can either turn a detection in $\tilde{m}$ to a \texttt{?} or a \texttt{?} to an incorrect value. Therefore, each insertion also adds one half-error to the Hamming distance between $m$ and $\tilde{m}$.
    \item Each incorrectly guessed symbol can also change up to two symbols in $\tilde{m}$ and therefore increase the Hamming distance between $m$ and $\tilde{m}$ by up to four.
\end{itemize}
This implies that the $m$ and $\tilde{m}$ are far apart by no more than $n(\delta+12\sqrt\eps)$. Having this reduction, we derive codes promised in \cref{thm:main-1} by taking the following near-MDS codes from \cite{guruswami2005linear} and indexing their codewords with a self-matching string.
\begin{theorem}[Guruswami and Indyk~{\cite[Theorem 3]{guruswami2005linear}}]\label{thm:guruswami-indyk-near-MDS-codes}
For every $r$, $0 < r < 1$, and all sufficiently small $\varepsilon > 0$, there exists an explicitly
specified family of GF(2)-linear (also called additive)
codes of rate $r$ and relative distance at least
$(1-r-\eps)$ over an alphabet of size $2^{O(\eps^{-4}r^{-1} \log(1/\eps))}$ such that codes from the family can be encoded
in linear time and can also be (uniquely) decoded in linear time from a fraction $e$ of errors and $s$
of erasures provided $2e + s \leq (1-r-\eps)$.
\end{theorem}

\begin{proof}[Proof of \cref{thm:main-1}]
Let $C$ be a code from \cref{thm:guruswami-indyk-near-MDS-codes} with relative distance $\delta_{{C}} = \delta + \frac{\eps}{3}$ and rate $1-\delta_C -\eps_{C}$ for $\eps_C=\frac{\eps}{3}$ and $S$ be an $\eps_S$-self-matching string with parameter $\eps_S = \frac{\eps^2}{36^2}$. We construct code ${C}'$ by simply taking the code ${C}$ and indexing each codeword of it with $S$. We claim that the resulting code satisfies the properties promised in the statement of the theorem.

We start with showing the decoding guarantee by describing the decoder. Note that a decoder can use the procedure described in \cref{alg:global-decoder} to use the index portion of codewords to reconstruct the codeword by up to a $\delta + 12\sqrt{\eps'} = \delta + \frac{\eps}{3} = \delta_{{C}}$ fraction of half-errors. The decoder then simply feeds the resulting string into the decoder of ${C}$ to fully recover the original string. The encoding and decoding complexities of ${C}'$ follow from the fact that ${C}$ is encodable and decodable in linear time and that \cref{alg:global-decoder} runs in quadratic time.

We finish the proof with proving the rate guarantee. Note that $\Sigma_{C'} = \Sigma_C \times \Sigma_S$.
\begin{eqnarray}
r_{C'} &=& \frac{|C'|}{n \log |\Sigma_{C'}|} 
= \frac{|C|}{n \log \left(|\Sigma_C|\times |\Sigma_S|\right)}\nonumber\\ 
&=& 
r_C \cdot \frac{\log |\Sigma_C|}{\log |\Sigma_C| + \log |\Sigma_S|} = \frac{r_C}{1 + \frac{\log |\Sigma_S|}{\log |\Sigma_C|}}\label{eqn:near-MDS-rate}
\end{eqnarray}
Note that the discussion in \cref{sec:pseudo-random-property} implies that $S$ can be over an alphabet of size $|\Sigma_S|=O(\eps^{-2})$ and \cref{thm:guruswami-indyk-near-MDS-codes} gives that $\log |\Sigma_C|=\omega(\eps^{-4}\log 1/\eps)$. Thus, $\frac{\log |\Sigma_S|}{\log |\Sigma_C|} = O(\eps^{-4})$, which plugged in \eqref{eqn:near-MDS-rate} implies that $r_{C'} \geq \frac{r_C}{1+O(\eps^{4})} \geq r_C - \frac{\eps}{3}$. Therefore, since the rate of code ${C}$ is $r_C = 1-\delta_{{C}}-\eps_{{C}} = 1-\delta-\frac{2}{3}\eps$, the rate of the code ${C}'$ is at least $r_{C'} = 1-\delta-\frac{2}{3}\eps - \frac{\eps}{3} = 1-\delta-\eps$.
\end{proof}

Note that the alphabet size of the codes from \cref{thm:main-1} is exponentially large in terms of $\eps^{-1}$. This is in sharp contrast to the Hamming error setting where there are codes known that can get $\eps$ close to unique decoding capacity with alphabets of polynomial size in terms of $1/\eps$. While large alphabet sizes might seem as an intrinsic weakness of the indexing-based code constructions, it turns out that an exponentially large alphabet size is actually necessary. We present the following theorem from~\cite{haeupler2018synchronization4} that shows any such code requires exponentially large alphabet size in terms of $\exp(\eps^{-1})$.


\begin{theorem}\label{cor:unique-alphabetSize}
There exists a function $f:(0,1) \to (0,1)$ such that for every $\delta,\eps>0$, every
family of insertion-deletion codes of rate $1-\delta-\eps$ that can be uniquely decoded from $\delta$-fraction of synchronization errors must have alphabet size $q \geq \exp\left(\frac{f(\delta)}{\eps}\right)$.
\end{theorem}
\begin{proof}[Proof Sketch]
For simplicity, assume that $\delta = \frac{d}{q}$ for some integer $d$. Consider a code of block length $n$ and an adversary that always deletes all occurrences of the $d$ least frequent symbols. With such an adversary, the string received on receiver's side will be a string of length $n(1-\delta)$ over an alphabet of size $q-d=q(1-\delta)$. This means that there are a total of 
$M={q \choose q-d}(q-d)^{n(1-\delta)}$ possible strings that may arrive at the receiver's end which implies that the rate of any such code is no more than
$$\frac{\log M}{n\log q}
=(1-\delta) \left(1+\frac{\log (1-\delta)}{\log q}\right)+o(1).$$
Therefore, to achieve a rate of $1-\delta-\eps$, 
\begin{eqnarray*}
&&1-\delta-\eps \leq (1-\delta) \left(1+\frac{\log (1-\delta)}{\log q}\right)\\
&\Rightarrow& (1-\delta)\frac{\log\frac{1}{1-\delta}}{\log q} \leq \eps 
\Rightarrow q \geq e^{(1-\delta)\log\frac{1}{1-\delta}}
\end{eqnarray*}
For the general case where $\delta q$ is not necessarily an integer, a similar, more careful argument proves the theorem. (See~\cite{haeupler2018synchronization4}.)
\end{proof}

The alphabet reduction idea used in the proof of \cref{cor:unique-alphabetSize} shows that deletions, in addition to reducing the information by eliminating symbols, reduce the information by essentially decreasing the information content of surviving symbols; suggesting that designating a part of each symbol to synchronization strings is not a waste of information.
A similar alphabet reduction argument is used in \cite{haeupler2018synchronization4,haeupler-list-dec-capacity2020} to derive strong upper-bounds on the zero-error list-decoding capacity of adversarial insertion-deletion channels.

\subsection{Near-Linear Time Codes}
In \cref{sec:technical-warm-up}, we presented a way to construct synchronization codes that approach the Singleton bound by taking a near-MDS error-correcting code and indexing its codewords with self-matching strings. 
In this section, we explain how the decoding complexity of such codes can be reduced to near-linear time.

The main idea is to replace the $\eps$-self matching string with one that satisfies a stronger pseudo-random property that allows for a near-linear time repositioning algorithm. We will thoroughly explain the construction of such a string and its repositioning algorithm in \cref{sec:edit-distance-approx-and-near-linear-repositioning}. We forward reference the properties of this string in the following theorem and defer the details to \cref{sec:edit-distance-approx-and-near-linear-repositioning}.

\begin{theorem}[\cref{lem:enhanced-sync-string-decoding} with $\eps_I=\frac{2\eps}{9}$, 
$\eps_s=\frac{\eps^2}{18}$, $K=\frac{6}{\eps}$, $\gamma=1$]\label{lem:enhanced-sync-string-summary}
For any $\eps >0$, there exist strings of any length $n$ over an alphabet of size $\exp\left(\frac{\log(1/\eps)}{\eps^3}\right)$ that, if used as an index in a synchronization channel with $\delta$ fraction of errors, enables a repositioning in $O_\eps(n\poly(\log n))$ time which guarantees no more than $n\eps$ incorrect guesses.
\end{theorem}

Using these strings in the code construction, the following can be achieved.

\begin{theorem}\label{thm:unique-decodable-codes}
For any $\eps>0$ and $\delta \in (0,1)$, and sufficiently large $n$, there exists an encoding map $E: \Sigma^k \rightarrow \Sigma^n$ and a decoding map $D: \Sigma^* \rightarrow \Sigma^k$, such that, if $\ED(E(m),x) \leq \delta n$ then $D(x) = m$. Further, $\frac{k}{n} > 1 - \delta - \eps$, $|\Sigma|=\exp\left(\eps^{-4}\log (1/\eps)\right)$, and $E$ and $D$ are explicit and can be computed in linear and near-linear time in terms of $n$ respectively. 
\end{theorem}
\begin{proof}[Proof Sketch]
%
To construct such codes with a given $\eps$, we use strings from \cref{lem:enhanced-sync-string-summary} with parameter $\frac{\eps}{4}$ as an index string. We then take code ${C}$ from \cite{guruswami2005linear} as a code with distance $\delta_{{C}} = \delta + \frac{\eps}{2}$ and rate $1-\delta_{{C}}-\frac{\eps}{4}$ over an alphabet of size $|\Sigma_{{C}}| \geq |\Sigma_S|^{4/\eps}$. Note that 
$|\Sigma_S|=\exp\left(\frac{\log (1/\eps)}{\eps^3}\right)$, therefore, the choice of $|\Sigma_{{C}}|$ is large enough to satisfy the requirements of \cite{guruswami2005linear}.
${C}$ is also encodable and decodable in linear time.

With this choice of ${C}$ and $S$, the same analysis as in \cref{sec:technical-warm-up} shows that the resulting synchronization code can be encoded in linear time, be decoded in $O_\eps(n\poly(\log n))$ time, corrects from any $\delta n$ insertions and deletions, achieves a rate of
$\frac{R_{{C}}}{1 + \frac{\log |\Sigma_S|}{\log |\Sigma_{{C}}|}}\ge \frac{1-\delta-3\eps/4}{1+\eps/4}\ge 1-\delta-\eps$, and is over an alphabet of size $\exp\left(\frac{\log (1/\eps)}{\eps^4}\right)$.
\end{proof}

\subsection{List Decoding: High Rate Codes}\label{sec:list-decodable-codes}

In this section, we review the results described in \cref{sec:main-result-list-dec} that, for every $0\leq \delta < 1$, every $0 \leq \gamma < \infty$ and every $\eps > 0$, gives list-decodable codes with rate $1 - \delta - \eps$, constant alphabet (so $q = O_{\delta,\gamma,\eps}(1)$), and sub-logarithmic list sizes. Furthermore, these codes are accompanied by efficient (polynomial time) decoding algorithms. We stress that the fraction of insertions can be arbitrarily large (more than 100\%), and the rate is independent of this parameter. Here is a formal statement of the result from \cite{haeupler2018synchronization4}.

\begin{theorem}\label{thm:list-decodable-code}
For every $0<\delta,\eps<1$ and $\gamma > 0$, there exist a family of list-decodable insertion-deletion codes that can protect against $\delta$-fraction of deletions and $\gamma$-fraction of insertions and achieves a rate of at least $1-\delta-\eps$ or more over an alphabet of size 
$\left(\frac{\gamma+1}{\eps^2}\right)^{O\left(\frac{\gamma+1}{\eps^3}\right)}=O_{\gamma, \eps}\left(1\right)$. These codes are list-decodable with lists of size $L_{\eps, \gamma}(n)= \exp\left(\exp\left(\exp\left(\log^* n\right)\right)\right)$, and have polynomial time encoding and decoding complexities.
\end{theorem}

The construction of these codes is similar to the ones from \cref{thm:unique-decodable-codes} except that the error-correcting code used in the construction is replaced with a high-rate list-recoverable code. A code ${C}$ given by the encoding function $\mathcal{E}:\Sigma^{nr}\rightarrow\Sigma^n$ is called to be $(\alpha, l, L)$-list recoverable if for any collection of $n$ sets $S_1, S_2, \ldots, S_n\subseteq\Sigma$ each of size $l$ or less, there are at most $L$ codewords for which more than $\alpha n$ elements appear in the list that corresponds to their position, i.e.,  $$\big|\left\{x \in {C}\mid \left|\left\{i \in [n] \mid x_i \in S_i\right\}\right| \geq \alpha n \right\}\big| \leq L.$$

The main idea is to use indexes and the repositioning algorithm from \cref{alg:global-decoder} to come up with a list of candidate symbols for each position of the original message and then feed these lists to the decoder of the list-recoverable code.
To prove \cref{thm:list-decodable-code}, the following family of list-recoverable codes from \cite{hemenway2017local} is utilized.

\begin{theorem}[Hemenway \emph{et al}.~{\cite[Theorem A.7]{hemenway2017local}}]\label{thm:PolyTimeListRecoverable}
Let $q$ be an even power of a prime, and choose $l, \epsilon > 0$, so that $q \geq \epsilon^{-2}$. Choose $\rho \in (0, 1)$. There is an $m_{min} = O(l \log_q(l/\epsilon)/\epsilon^2)$ so that the following holds for all $m \geq m_{min}$. For infinitely many $n$ (all $n$ of the form $q^{e/2}(\sqrt q - 1)$ for any integer $e$), there is a deterministic polynomial-time construction of an $F_q$-linear code $C:\mathbb{F}^{\rho n}_{q^m} \rightarrow \mathbb{F}^n_{q^m}$ of rate $\rho$ and relative distance $1-\rho - O(\epsilon)$ that is $(1 - \rho - \epsilon, l, L)$-list-recoverable in time $\poly(n, L)$, returning a list of codewords that are all contained in a subspace over $\mathbb{F}_q$ of dimension at most 
$\left(\frac{l}{\epsilon}\right)^{2^{\log^*(mn)}}$; implying that $L \leq q^{\left(l/\epsilon\right)^{2^{\log^*(mn)}}}$.
\end{theorem}
\begin{proof}[Proof of \cref{thm:list-decodable-code}]
By setting parameters $\rho=1-\delta-\frac{\eps}{2}$, $l=\frac{2(\gamma+1)}{\eps}$, and $\epsilon=\frac{\eps}{4}$ in Theorem~\ref{thm:PolyTimeListRecoverable}, one can obtain a family of codes $\mathcal{C}$ that achieves rate $\rho=1-\delta-\frac{\eps}{2}$ and is $(\alpha, l, L)$-recoverable in polynomial time for $\alpha=1-\delta-\eps/4$ and some $L=\exp\left(\exp\left(\exp\left(\log^*n\right)\right)\right)$ (by treating $\gamma$ and $\eps$ as constants). 
Such a family of codes can be found over an alphabet $\Sigma_{\mathcal{C}}$ of size $q=\left(l/\epsilon\right)^{O(l /\epsilon^2)}= \left(\frac{\gamma+1}{\eps^2}\right)^{O\left(\frac{\gamma+1}{\eps^3}\right)} = O_{\gamma, \eps}(1)$ or infinitely many integer numbers larger than $q$.

We index the codewords of this code with an $\eps_s=\frac{\eps^2}{64(1+\gamma)}$ self-matching string $S$. We now show that these codes satisfy the list-decoding properties presented in the statement of the theorem.

The decoder starts with guessing the positions for the symbols it receives using a repositioning algorithm similar to \cref{alg:global-decoder} with two minor differences:
\begin{enumerate}
    \item Instead of reconstructing the original string with guessed positions and \texttt{?}s, the decoder compiles a list for each position containing all received symbols that have been guessed to be in that position.
    \item The algorithm repeats the procedure of calculating $\LCS$ and adding elements to the lists for a total of $K=\frac{8(1+\gamma)}{\eps}$ times.
\end{enumerate}
A similar analysis to the one for \cref{alg:global-decoder} shows that the count of the lists that do not contain the original symbol that corresponds to their position is no more than
$n\left(\delta + \frac{1+\gamma}{K} + K\eps_s\right)=n(\delta+\eps/4).$

Then, the decoding algorithm feeds these lists into the decoder of the list-recoverable code from $\mathcal{C}$ to obtain a list of size $L$ of potential original messages. Since the parameter $\alpha$ was chosen to be $1-\delta-\eps/4$, the output list is guaranteed to contain the original message.

The rate of the resulting family of codes is $\frac{1-\delta-\eps/2}{1+\log\left|\Sigma_S\right|/\log \left|\Sigma_{\mathcal{C}}\right|}$
which, by taking $|\Sigma_{\mathcal{C}}|$ large enough in terms of $\eps$, is larger than $1-\delta-\eps$. As $\mathcal{C}$ is encodable and decodable in polynomial time, the encoding and decoding complexities of the indexed code will be polynomial as well.
\end{proof}

We remark that the self-matching string in the construction of list-decodable synchronization codes from~\cref{thm:list-decodable-code} can be replaced with the near-linear time repositionable indexes of  \cref{lem:enhanced-sync-string-summary}. This would reduce the time complexity of the repositioning subroutine in the decoding algorithm to near-linear time. Therefore, it would allow one to improve the decoding complexity of these codes upon discovery of high-rate list-recoverable codes with faster decoders, potentially to near-linear time. A recent work of Kopparty \emph{et al}.~\cite{kopparty2019list} has broken this barrier and offers list-recoverable tensor codes with a deterministic $n^{1+o(1)}$ time decoding. Using such codes in a similar scheme would yield a near-linear time list-decodable family of codes with similar properties as of \cref{thm:list-decodable-code} albeit over alphabet sizes that grow in terms of the block length of the code.

\medskip

\subsection{List Decoding: Optimal Resilience via Concatenation}
In this section, we discuss the result presented in \cref{sec:main-result-resilience} that fully characterizes error resilience for list-decodable synchronization codes. More precisely, \cite{guruswami2019optimally} exactly identifies the maximal fraction of insertion and deletion errors tolerable by $q$-ary list-decodable codes with non-vanishing rate.

We start by describing the result for binary codes. Note that no positive-rate code can be list-decoded from a $\delta=\frac{1}{2}$ fraction of deletions as an adversary can simply delete all instances of the less frequent symbol. Similarly, no positive-rate code can be list-decoded from a fraction $\gamma=1$ of insertions since any string of length $n$ can be turned into $(01)^n$ with $n$ insertions. A simple time-sharing argument would show that an adversary that can apply any combination of $\delta$ fraction of deletions and $\gamma$-fraction of insertions that satisfy $\gamma + 2\delta = 1$ can make the list-decoding impossible. The following theorem from \cite{guruswami2019optimally} shows the existence of positive-rate list-decodable codes otherwise.

\begin{theorem}\label{thm:resilience-binary}
For any $\eps\in(0, 1)$ and sufficiently large $n$, there exists a constant-rate family of efficient binary codes that are $L$-list decodable from any $\delta n$ deletions and $\gamma n$ insertions in $\poly(n)$ time as long as $\gamma + 2\delta \le 1-\eps$ where $n$ denotes the block length of the code, $L=O_\eps(\exp(\exp(\exp(\log^*n))))$, and the code achieves a rate of $\exp\left(-\frac{1}{\eps^{10}}\log^2 \frac{1}{\eps}\right)$.
\end{theorem}

This result is generalized for larger alphabets in \cite{guruswami2019optimally}. However, the feasibility region for larger alphabet sizes is more complex. We start with showing that list-decoding is impossible for several points $(\gamma, \delta)$ that lie on a quadratic curve. This implies a piece-wise linear outer-bound for the resilience region.

\begin{theorem}\label{thm:counterexamples}
For any alphabet size $q$ and any $i=1, 2, \cdots, q-1$, no positive-rate $q$-ary infinite family of insertion-deletion codes can list-decode from $\delta = \frac{q-i}{q}$ fraction of deletions and $\gamma = \frac{i(i-1)}{q}$ fraction of insertions.
\end{theorem}
\begin{proof}
Take a codeword $x\in [q]^n$. With $\delta n = \frac{q-i}{q} \cdot n$, the adversary can delete the $q-i$ least frequent symbols to turn $x$ into $x'\in\Sigma_d^{n(1-\delta)}$ for some $\Sigma_d = \{\sigma_1, \cdots, \sigma_{i}\} \subseteq [q]$. Then, with $\gamma n = n(1-\delta)(i - 1)$ insertions, it can turn $x'$ into $[\sigma_1, \sigma_2, \cdots, \sigma_{i}]^{n(1-\delta)}$, i.e., $n(1-\delta)$ repetitions of the string $\sigma_1, \sigma_2, \cdots, \sigma_{i}$. Such an adversary only allows $O(1)$ amount of information to pass to the receiver. Hence, no such family of codes can yield a positive rate.
\end{proof}

It is shown in \cite{guruswami2019optimally} that this is indeed the error resilience region for list-decoding.

\begin{theorem}\label{thm:resilience-qAry}
For any positive integer $q \geq 2$, let $F_q$ be defined as the concave polygon over vertices
$\left(\frac{i(i-1)}{q}, \frac{q-i}{q}\right)$ for $i = 1, \cdots, q$ and $(0, 0)$. (An illustration for $q=5$ is presented in \cref{fig:actual-region}). $F_q$ does not include the border except the two segments 
$\left[(0, 0), (q-1, 0)\right)$ and $\left[(0, 0), \left(0, 1-1/q\right)\right)$. 
Then, for any $\eps > 0$ and sufficiently large $n$, there exists a family of $q$-ary codes that, as long as $(\gamma, \delta) \in (1-\eps)F_q$, are efficiently $L$-list decodable from any $\delta n$ deletions and $\gamma n$ insertions where $n$ denotes the block length of the code, $L=O(\exp(\exp(\exp(\log^*n))))$, and the code achieves a positive rate of $\exp\left(-\frac{1}{\eps^{10}}\log^2 \frac{1}{\eps}\right)$.
\end{theorem}
\begin{figure}[]
  \IEEEOnly{\setlength\abovecaptionskip{0pt}}
  \centering
  \includegraphics[width=\linewidth]{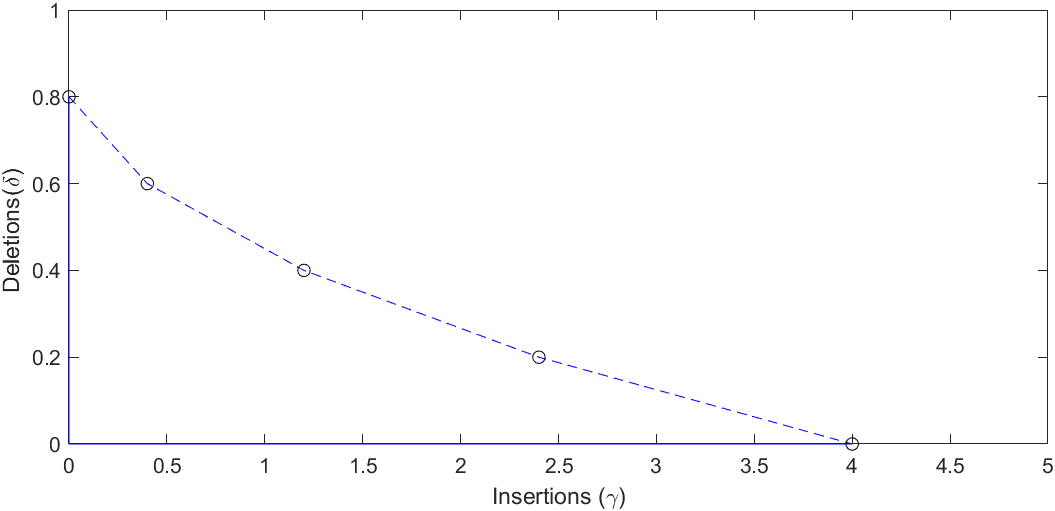}\\
  \caption{Feasibility region for $q=5$.}\label{fig:actual-region}
\end{figure}

We give a high-level description of the steps taken in the proof of \cref{thm:resilience-binary}. The proof of \cref{thm:resilience-qAry} follows the same blueprint but is more technically involved. For a formal proof of both theorems, we refer the reader to \cite{guruswami2019optimally}.

\cref{thm:resilience-binary} is achieved via two main ingredients. The first is a simple new concatenation scheme for list-decodable synchronization codes which can be used to boost the rate of insdel codes. The second component is a technically intricate proof of the list-decoding properties of the Bukh-Ma codes~\cite{bukh2014longest} which have good distance properties but a small sub-constant rate.

\medskip
\subsubsection*{(I) Concatenating List-Decodable Synchronization Codes}
The first ingredient is a simple but powerful framework for constructing list-decodable insertion-deletion codes via code concatenation. Recall that code concatenation comprises the encoding of an \emph{outer} code $C_{\rm out}$ with an \emph{inner} code $C_{\rm in}$ whose size equals the alphabet size of $C_{\rm out}$.

In this approach, the outer code $C_{\rm out}$ is chosen to be a list-decodable insdel code $C_{\rm out}$ over an alphabet whose size is some large function of $1/\eps$, has a constant rate, and is capable of tolerating a huge number of insertions. 
Such a code is introduced in~\cite{haeupler2018synchronization4} and presented in \cref{thm:list-decodable-code}.

The inner code $C_{\rm in}$ is chosen to be a list-decodable insdel code over the binary alphabet (or desired alphabet of size $q$ for \cref{thm:resilience-qAry}), which has non-trivial list decoding properties for the desired fractions $\gamma$ and $\delta$ of insertions and deletions. Most notably, the concatenation framework requires the inner code to be chosen from a family of good list-decodable insdel codes \emph{with an arbitrarily large number of codewords}, and a list-size bounded by some fixed function of $1/\eps$. The codes of Bukh and Ma~\cite{bukh2014longest} are shown to satisfy these properties and are used in \cite{guruswami2019optimally} as the inner code.

We show that even if $C_{\rm in}$ has an essentially arbitrarily bad sub-constant rate and is not efficient, the resulting $q$-ary insdel code does have constant rate, and can also be efficiently list-decoded from the same fraction of insertions and deletions as $C_{\rm in}$. For the problem considered in this paper, this framework essentially provides efficiency of codes for free. 

The encoding is straightforwardly done by the standard concatenation procedure. The decoding procedure on the other hand, is considerably simpler than similar schemes introduced in earlier works~\cite{schulman1999asymptotically,guruswami2016efficiently,guruswami2017deletion,bukh2017improved}. The decoding is done by (i) list-decoding a sliding substring of the received string using the inner code $C_{\rm in}$, (ii) creating a single string from the symbols in these lists, and (iii) using the list-decoding algorithm of the outer code on this string (viewed as a version of the outer codeword with some number of deletions and insertions). 

The main driving force behind why this simplistic sounding approach actually works is a judicious choice of the outer code $C_{\rm out}$. Specifically, these codes can tolerate a very large number of insertions. This means that the many extra symbols coming from the list-decodings of the inner code $C_{\rm in}$ and the choice of the (overlapping) sliding intervals does not disrupt the decoding of the outer code. Further, as mentioned above, the list size of the inner code only depends on $\eps$ and is independent of the size of the code. This is a crucial property for this concatenation scheme as the following order is used to choose the parameters of $C_{\rm in}$ and $C_{\rm out}$. Having the parameter $\eps$, the fraction of insertions that the outer code needs to protect against is determined. This would dictate the alphabet size of the outer code and subsequently, the block length or the size of the inner code.

\medskip

\subsubsection*{(II) Analyzing the List-Decoding Properties of Bukh-Ma Codes}
For the inner code in the concatenation scheme described above, we use a simple family of codes introduced by Bukh and Ma~\cite{bukh2014longest}, which consist of strings $(0^r\ 1^r)^{\frac{n}{2r}}$ that oscillate between $0$s and $1$s with different frequencies. (Below we will refer to $r$ as the \emph{period}, and $1/r$ should be thought of as the \emph{frequency} of alternation.) It is shown in \cite{guruswami2019optimally} that these codes satisfy the following properties.
\begin{theorem}\label{thm:LowRateListDecBinaryCodes}
For any $\eps>0$ and sufficiently large $n$, let $C_{n, \eps}$ be the following Bukh-Ma code:
$${C}_{n, \eps} = \left\{\left(0^r1^r\right)^{\frac{n}{2r}}\Big| r=\left(\frac{1}{\eps^4}\right)^k, k< \log_{1/\eps^4} n\right\}.$$
For any $\delta,\gamma \geq 0$ where $\gamma+2\delta < 1-\eps$, it holds that ${C}_{n, \eps}$ is list decodable from any $\delta n$ deletions and $\gamma n$ insertions with a list size of $O\left(\frac{1}{\eps^3}\right)$.
\end{theorem}

In order to prove \cref{thm:LowRateListDecBinaryCodes}, \cite{guruswami2019optimally} first introduces a new correlation measure which expresses how close a string is to any given frequency (or Bukh-Ma codeword). 
Using this measure, we want to show that it is impossible to have a single string $v$ which is more than $\eps$-correlated with more than $\Theta_\eps(1)$ frequencies. 
Loosely speaking, the parameter $\eps$ in $\eps$-correlation indicates how close the term $2D+I$ is to $1$ where $D$ denotes the number of deletions and $I$ denotes the number of insertions required to convert the string $v$ into some Bukh-Ma codeword when picking the set of insertions and deletions that minimizes $2D+I$.

The proof technique utilized by \cite{guruswami2019optimally} is somewhat reminiscent of the one used to establish the polarization of the martingale of entropies in the analysis of polar codes~\cite{arikan2008channel,blasiok2018polar}. 
In more detail, \cite{guruswami2019optimally} recursively sub-samples smaller and smaller nested substrings of $v$, and analyzes the expectation and variance of the bias between the fraction of $0$'s and $1$'s in these substrings. More precisely, it orders the run lengths $r_1,r_2,\ldots$ that are $\eps$-correlated with $v$ in decreasing order and first samples a substring $v_1$ with $r_1 \gg |v_1| \gg r_2$ from $v$. While the expected zero-one bias in $v_1$ is the same as in $v$, \cite{guruswami2019optimally} shows that the variance of this bias is a strictly increasing function in the correlation with $\left(0^{r_1} 1^{r_1}\right)^{\frac{n}{2{r_1}}}$. Intuitively, $v$ cannot be too uniform on a scale of length $|v_1|$ if it is correlated with $r_1$. 

In other words, if $v$ is $\eps$-correlated with $r_1$, the sampled substring $v_1$ will land in a part of $v$ which is either similar to one of the long stretches of zeros in $v$ or in a part which is similar to a long stretch of ones in $v$, resulting in some positive variance in the bias of $v_1$. Furthermore, because the scales $r_2, r_3, \ldots$ are so much smaller than $v_1$, this sub-sampling of $v_1$ preserves the correlation with these scales intact, at least in expectation.

Next, a substring $v_2$ with $r_2 \gg |v_2| \gg r_3$ is sampled within $v_1$. Again, the bias in $v_2$ stays the same as the one in $v_1$ in expectation but the sub-sampling introduces even more variance given that $v_1$ is still non-trivially correlated with the string with period $r_2$. The evolution of the bias of the strings $v_1, v_2, \ldots$ produced by this nested sampling procedure can now be seen as a martingale with the same expectation but an ever increasing variance. Given that the bias is bounded in magnitude by 1, the increase in variance cannot continue indefinitely. This limits the number of frequencies a string $v$ can be non-trivially correlated with and, subsequently, implies the list-decodability property of the code.

\section{Synchronization Strings}\label{sec:sync-strings}
In this section, we discuss synchronization strings introduced in \cite{haeupler2017synchronization} and review their combinatorial properties and applications. We also overview extensions and enhancements made to synchronization strings from \cite{cheng2018synchronization,haeupler2017synchronization2,haeupler2017synchronization3,haeupler2018synchronization4,haeupler2019near}.

\begin{definition}[$\eps$-synchronization strings]
String $S \in \Sigma^n$ is an $\eps$-synchronization string if for every $1 \leq i < j < k \leq n + 1$ we have that $\ED\left(S[i, j),S[j, k)\right) > (1-\eps) (k-i)$. 
\end{definition}

In this definition, $\ED$ represents the edit distance function and $S[x, y)$ denotes a substring of $S$ starting from position $x$ and ending at position $y-1$. We use similar notations $S[x, y]$, $S(x, y]$, and $S(x, y)$ to denote substrings of $S$ where $(,)$ and $[,]$ denote the exclusion and inclusion of the starting/end point of the interval respectively.

In simpler terms, the $\eps$-synchronization property is a pseudo-random property that requires all neighboring substrings of the string to be far apart under the edit distance metric. It is shown in \cite{haeupler2017synchronization} that $\eps$-synchronization is not only a strictly stronger property than the $\eps$-self-matching property but also a hereditary extension of it. More precisely, if all substrings of a string satisfy the $\frac{\eps}{2}$-self matching string property, then the string itself is an $\eps$-synchronization string.

\subsection{Existence}\label{sec:sync-strings-existence}
It is shown in \cite{haeupler2017synchronization} that, similar to self matching strings, arbitrarily long $\eps$-synchronization strings exist over alphabets whose size is independent of the string length. More precisely, \cite{haeupler2017synchronization} shows the existence of arbitrarily long strings over an alphabet of size $O(\eps^{-2})$ that satisfy the $\eps$-synchronization property for pairs of neighboring substrings of total length $\frac{1}{\eps^2}$ or more by utilizing the Lov\'{a}sz's local lemma to show that the probability of such an event for a random string is non-zero. Indexing such a string with a string formed by repetitions of $1, 2, \cdots, \eps^{-2}$ ensures the $\eps$-synchronization property over smaller substrings and gives an $\eps$-synchronization string over an alphabet of $O(\eps^{-4})$ size. With a non-uniform sample space, \cite{cheng2018synchronization} utilizes the Lov\'{a}sz's local lemma in the same manner to reduce the alphabet size to $O(\eps^{-2})$ and gives the following.

\begin{theorem}\label{thm:existence}
For any $\eps\in (0,1)$, there exists an alphabet $\Sigma$ of size $O(\eps^{-2})$ so that for any $n \geq 1$, there exists an $\eps$-synchronization string of length $n$ over $\Sigma$.
\end{theorem}

\subsubsection*{Extremal Properties}
We would like to add a brief remark regarding extremal questions that are raised by the definition of the synchronization string property. One interesting question is what is the minimal function of $\eps$ as alphabet size for which \cref{thm:existence} holds. It has been shown in \cite{cheng2018synchronization} that any such alphabet has to be of size $\Omega(\eps^{-3/2})$. This leaves us with the open question of where the minimal alphabet size lies between $\Omega(\eps^{-3/2})$ and $O(\eps^{-2})$.

A similar question can be asked for non-specific values of $\eps$, i.e., what is the smallest alphabet size over which arbitrarily long $\eps$-synchronization strings exist for any $\eps < 1$. It is easy to observe that any binary string of length 4 or more contains two identical neighboring substrings. Also, it has been shown that arbitrarily long $\frac{11}{12}$-synchronization strings exist over an alphabet of size four~\cite{cheng2018synchronization}. This leaves the open question of whether arbitrarily long synchronization strings exist over a ternary alphabet or not.

\subsection{Online Decoding for Synchronization Strings}\label{sec:sync-strings-and-online-decoding}
In this chapter, we introduce an online repositioning algorithm for synchronization strings. In the same spirit as \cref{sec:technical-warm-up}, we show that synchronization strings can be used to guess the original position of symbols undergoing insertion-deletion errors via indexing. 
However, for synchronization strings, the repositioning can be done in an online fashion, i.e., the position of each symbol is guessed upon its arrival and without waiting for the rest of the communication to take place. This enables a delay-free simulation of a channel with Hamming-type errors over any given insertion-deletion channel with adequately large alphabet size. We will discuss this further in \cref{sec:other-applications}.

To present the online repositioning algorithm, we introduce the notion of relative suffix distance inspired by a similar notion from \cite{braverman2015coding}.
\begin{definition}[Relative Suffix Distance]
For any $S, S' \in \Sigma^*$, their relative suffix distance ($\RSD$) is defined as follows:
$$\RSD(S,S') = \max_{k > 0} \frac{\ED\Big(S\big(|S|-k,|S|\big],S'\big(|S'|-k,|S'|\big]\Big)}{2k}$$
\end{definition}
It is shown in \cite{haeupler2017synchronization} that $\RSD$ is a metric that takes a value within $[0, 1]$. The interesting property of $\RSD$ that comes in handy when devising an online repositioning algorithm is that the prefixes of a synchronization string are far apart under the $\RSD$ metric.
\begin{proposition}\label{lem:RSD-prefixes}
Let $S$ be an $\eps$-synchronization string. For any $i\neq j$, $\RSD(S[1, i], S[1, j]) > 1-\eps$.
\end{proposition}

Note that an online repositioning algorithm is essentially one that decides which prefix of the message string is sent upon arrival of each symbol at the receiver side. Therefore, the online repositioning algorithm only needs to decide which prefix of the synchronization string is the most consistent to the index portion of the received string up until the arrival of each symbol. To this end, \cref{lem:RSD-prefixes} suggests the natural repositioning strategy of finding the closest prefix of the utilized synchronization string to the index part of the received string under relative suffix distance and declaring the length of that prefix as the position of that symbol.

The guarantees that this decoding strategy provides is discussed in details in \cite{haeupler2017synchronization}. However, we remark that the suffix distance between a string $s$ and a noisy version of it, $\tilde{s}$, that is altered by insertions and deletions is particularly sensitive to how dense the fraction of error occurrences is in small suffixes of $\tilde{s}$. This implies that occurrences of insertions and deletions can only disrupt the correctness of this repositioning strategy for some of the following symbols and the effect would fade away as communication goes on. By formalizing these observations and employing a similar yet more complicated distance function, \cite{haeupler2017synchronization} gives the following.

\begin{theorem}\label{thm:online-reposition-misdecodings}
There exists an online repositioning algorithm for a communication of length $n$ over a channel with up to $n\delta$ synchronization errors that, assuming that the message is indexed by an $\eps$-synchronization string, guesses the position of each received symbol in $O(n^4)$ time and incorrectly guesses the positions of no more than $\frac{n\delta}{1-\eps}$ received symbols.
\end{theorem}

Note that, as opposed to the repositioning algorithm in \cref{sec:technical-warm-up}, the number of incorrect guesses does not tend to zero by taking smaller values for $\eps$. In fact, if one constructs synchronization codes as in \cref{sec:technical-warm-up} with $\eps$-synchronization strings and uses this repositioning algorithm instead of \cref{alg:global-decoder}, the rate achieved is $1-3\delta-\eps^{O(1)}$.

\subsection{Construction: Long-Distance Synchronization Strings}
To construct synchronization strings, \cite{haeupler2017synchronization3} utilizes the algorithmic  Lov\'{a}sz local lemma of Chandrasekaran \emph{et al}. \cite{chandrasekaran2013deterministic} with a similar random space to the one used in \cref{sec:sync-strings-existence} and obtains an efficient construction of such strings over an alphabet of size $O(\eps^{-4})$. In this section, we review the steps taken in \cite{haeupler2017synchronization3} to obtain a linear-time explicit construction for synchronization strings. In order to do so, we start with presenting the \emph{long-distance synchronization string} property that generalizes the requirement of large edit distance to non-adjacent substrings that are at least logarithmically long in terms of the length of the string.

\begin{definition}[$c$-long-distance $\eps$-synchronization string]\label{def:long-distance}
String $S \in \Sigma^n$ is a $c$-long-distance $\eps$-synchronization string if for every pair of substrings $S[i, j)$ and $S[i', j')$ that are either adjacent or of total length $c\log n$ or more, $ED\left(S[i,j),S[i',j')\right) > (1-\eps) l$ where $l=j-i+j'-i'$.
\end{definition}
We now describe construction algorithms for (long-distance) synchronization strings.

\medskip

\subsubsection{Boosting Step I: Linear Time Construction}\label{sec:linear-time-construction}
\cite{haeupler2017synchronization3} provides a simple boosting step which allows a polynomial speed-up to any synchronization string construction at the cost of increasing the alphabet size by proposing a construction of an $O(\eps)$-synchronization string of length $O_\eps(n^2)$ having an $\eps$-synchronization string of length $n$. 

\begin{lemma}\label{lem:simplepolyboosting}
Fix an even $n \in \mathbb{N}$ and $\gamma > 0$ such that $\gamma n \in \mathbb{N}$. Suppose $S \in \Sigma^n$ is an $\eps$-synchronization string. The string $S' \in \Sigma'^{\gamma n^2}$ with $\Sigma' = \Sigma^3$ and 
\begin{equation}
S'[i] = \left(S[i \bmod n], S[(i+n/2) \bmod n], S\left[\left\lceil {\frac{i}{\gamma n}}\right\rceil\right]\right)\label{eqn:boosting-1}
\end{equation}
 is an $(\eps + 6\gamma)$-synchronization string of length $\gamma n^2$. 
\end{lemma}
\begin{proof}[Proof Sketch]
$S'$ is formed by the symbol-wise concatenation of three strings as presented in \cref{eqn:boosting-1}. The first two elements form repetitions of $S$ which guarantee the synchronization property over small intervals and the third element that guarantees the synchronization property over larger intervals.
\end{proof}

Employing this boosting technique for an adequately large number of times can turn the polynomial-time construction of synchronization strings obtained by the algorithmic Lov\'{a}sz local lemma of \cite{chandrasekaran2013deterministic} into a linear time construction at the cost of a larger alphabet that is still of $\eps^{-O(1)}$ size.

\medskip

\subsubsection{Boosting Step II: Explicit Linear-Time Long-Distance Construction}\label{sec:boosting-2}
We now describe a second boosting step introduced in \cite{haeupler2017synchronization3} that takes the linear-time construction from the previous section and turns it into a linear-time construction for long-distance synchronization strings that is also highly-explicit, i.e., for any index $i$, it can compute the substring $[i, i+\log n]$ in $O(\log n)$ time.

To describe the construction, we first point out a connection between long-distance synchronization strings and synchronization codes. Note that if one splits a $c$-long-distance $\eps$-synchronization string into substrings of length $c\log n$, the long-distance synchronization property will require that any pair of resulting substrings to have an edit distance of at least $2(1-\eps)c\log n$, i.e., form an insertion-deletion code of relative distance $1-\eps$. 

Similarly, given a synchronization code $C$ of distance $1-\eps$, rate $r>0$ and block length $N$, one can construct a string of length $n=\exp(Nr)$ by appending the codewords of $C$ together which satisfies the $c$-long-distance $\eps'$-synchronization property for pairs of substrings of total length $\Omega(\log n)$ where $\eps'=O(\eps)$ and $c = O_{\eps}(1)$. This claim is proved by a simple combinatorial argument using the distance property of the code $C$. We refer the reader to \cite{haeupler2017synchronization3} for a formal proof.

The second boosting step uses this observation and makes a long-distance synchronization string by symbol-wise concatenation of the string described above with a string that guarantees the $\eps$-synchronization property for neighboring intervals of total length $O(\log n)$. More formally, having the code $C$ of block length $N$ the construction is as follows.
\begin{equation}
S[i]=\Bigg(C\left(\left\lfloor \frac{i}{N}\right\rfloor\right) \left[i \left( \bmod\, N\right)\right], T[i]\Bigg),\label{eq:boosting-2}
\end{equation}
where $T$ is the symbol-wise concatenation of two shifted repetitions of some synchronization string $S'$ of length $O(\log n)$, i.e., 
$T[i] = \left(S'[i \bmod l], S'[(i+L/2) \bmod l]\right)$ for $l=|S'|$.
\begin{figure}
\centering
\includegraphics[width=\linewidth]{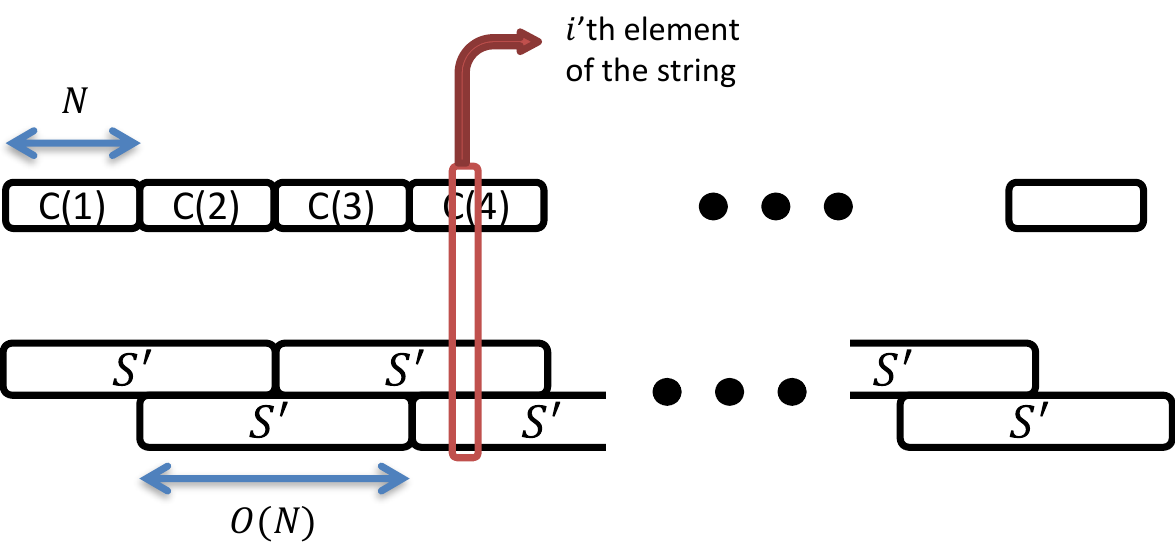}
\IEEEOnly{\setlength\abovecaptionskip{0pt}}
\caption{Pictorial representation of the construction of a long-distance $\eps$-synchronization string of length $n$.}
\label{fig:boosting-2}
\end{figure}
A pictorial description of the construction is available at \cref{fig:boosting-2}. Given that the codewords of the code $C$ take care of the long-distance synchronization property for longer pairs of intervals and repetitions of $S'$ provide the synchronization string guarantee for short neighboring intervals, this construction yields a long-distance synchronization string. In order to show the linear-time and highly-explicit construction, the following two ingredients are necessary:
\begin{itemize}
    \item A linear time construction for synchronization string $S$. This is provided by boosting step I.
    \item Linear time construction for code $C$. To obtain this, a family of high-rate synchronization codes with linear-time construction is necessary. To obtain such a family of codes, \cite{haeupler2017synchronization3} takes the near-MDS code of Guruswami and Indyk~\cite{guruswami2005linear} and concatenates it with an inner code to reduce its alphabet to $\eps^{O(1)}$. Note that the size of the inner code is $O_\eps(1)$. Therefore, the encoding time of the resulting family of codes remains linear and its rate is still positive. \cite{haeupler2017synchronization3} then indexes the codewords of this code with an $\eps$-self matching string as in \cref{sec:technical-warm-up} to obtain the necessary synchronization code for this construction.
\end{itemize}

The details of this construction are available in \cite{haeupler2017synchronization3}. We summarize the guarantees of the construction in the following theorem. 
\begin{theorem}\label{thm:explicit-construction}
There is a deterministic algorithm that, for any constant $0< \eps  < 1$ and $n \in \mathbb{N}$, computes a $c=\eps^{-O(1)}$-long-distance $\eps$-synchronization string $S \in \Sigma^n$ where $|\Sigma|=\eps^{-O(1)}$. This construction runs in linear time and, moreover, any substring $S[i, i+\log n]$ can be computed in $O_\eps(\log n)$ time.
\end{theorem}

\subsection{Infinite Synchronization Strings}
An infinite $\eps$-synchronization string is naturally defined as an infinite string, in which, any two neighboring intervals $[i ,j)$ and $[j, k)$ have an edit distance of at least $(1-\eps)(k-i)$.
Existence of infinite $\eps$-synchronization strings can be proved via a simple topological argument.
Fix any $\eps \in (0, 1)$. According to Theorem~\ref{thm:existence} there exist an alphabet $\Sigma$ of size $O(1/\eps^2)$ such that there exists at least one $\eps$-synchronization strings over $\Sigma$ for every length $n \in \mathbb{N}$. We will define an infinite synchronization string $S = s_1 \cdot s_2 \cdot s_3 \cdots$ with $s_i \in \Sigma$ for any $i \in \mathbb{N}$ inductively. We fix an ordering on $\Sigma$ and define $s_1 \in \Sigma$ to be the first symbol in this ordering such that an infinite number of these strings start with $s_1$. Given that there is an infinite number of $\eps$-synchronization strings over $\Sigma$, such an $s_1$ exists. Furthermore, the subset of $\eps$-synchronization strings over $\Sigma$ which start with $s_1$ is infinite by definition, allowing us to define $s_2 \in \Sigma$ to be the lexicographically first symbol in $\Sigma$ such there exists an infinite number of $\eps$-synchronization strings over $\Sigma$ starting with $s_1 \cdot s_2$. In the same manner, we inductively define the whole string. Since each prefix of this string satisfies the $\eps$-synchronization property, all pairs of adjacent intervals satisfy the $\eps$-synchronization property and this whole string is indeed an infinite $\eps$-synchronization string.

The construction from \cref{thm:explicit-construction} can be generalized for infinite synchronization strings as follows.
\begin{theorem}\label{thm:infinite-construction}
For all $0 < \eps < 1$, there exists an infinite $\eps$-synchronization string $S$ over a $\poly(\eps^{-1})$-sized alphabet so that any prefix of it can be computed in linear time. Further, for any $i$, $S[i, i+\log i]$ can be computed in $O(\log i)$ time.
\end{theorem}
The proof of \cref{thm:infinite-construction} utilizes a construction for infinite synchronization strings obtained by concatenation of finite synchronization strings of exponentially increasing length. More precisely, let $S_i$
 denote a $\Theta(\eps)$-synchronization string of length $i$. Further, let $U$ and $V$ be as follows:
 $$U = (S_k, S_{k^3}, S_{k^5}, \dots), \qquad V = (S_{k^2}, S_{k^4}, S_{k^6}, \dots)$$
Then \cite{haeupler2017synchronization3} shows that the symbols-wise concatenation of these strings, i.e., string $T$ where $T[i] = (U[i], V[i])$ is an infinite synchronization string. A pictorial representation of the construction of $T$ is available in \cref{fig:infinite-construction}. The proof of \cref{thm:infinite-construction} is derived by simply using the above-mentioned construction and utilizing \cref{thm:explicit-construction} to construct the finite strings used to form $U$ and $V$.
\begin{figure}
\centering
\includegraphics[width=\linewidth]{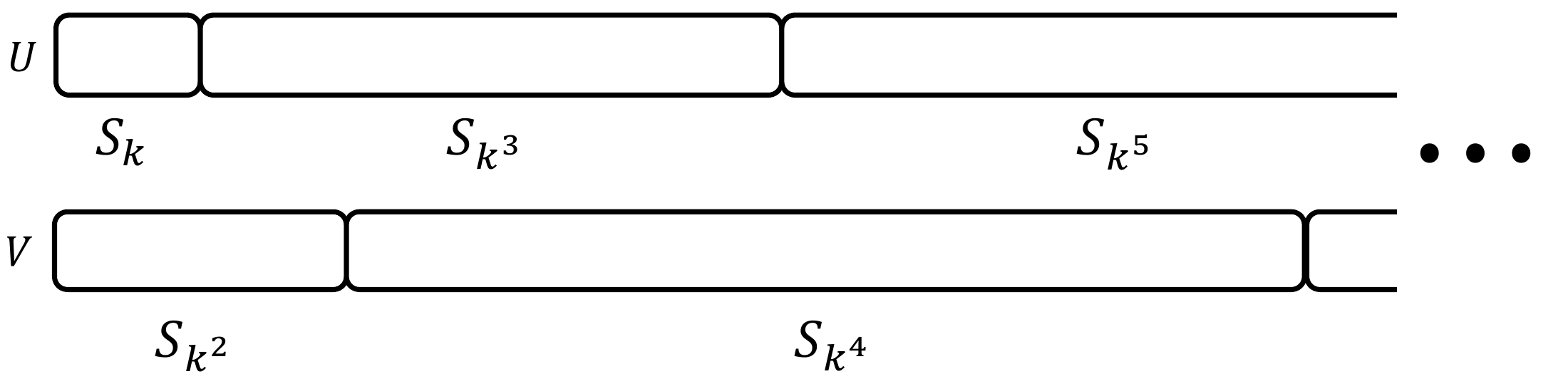}
  \IEEEOnly{\setlength\abovecaptionskip{0pt}}
\caption{Construction of infinite synchronization string $T$}
\label{fig:infinite-construction}
\end{figure}

\subsection{Local Decoding for Long-Distance Synchronization Strings}\label{sec:local-decoding}
In this section, with a slight modification to the construction from \eqref{eq:boosting-2} for long-distance synchronization strings, we give an index sequence which facilitates local repositioning. A \emph{local} repositioning algorithm is one that guesses the position of a received symbol using only the knowledge of a small $O(\log n)$-sized neighborhood of the surrounding received symbols, as opposed to all received symbols which is what \cref{alg:global-decoder} does. 

Consider the following indexing sequence that is obtained from the construction of \eqref{eq:boosting-2} by concatenating an extra term that essentially works as a circular index counter for insertion-deletion code blocks.
\begin{equation}
R[i]=\Bigg(C\left(\left\lfloor \frac{i}{N}\right\rfloor\right) \left[i \left( \bmod\, N\right)\right], T[i], \left\lfloor \frac{i}{N}\right\rfloor \left(\bmod\, \frac{8}{\eps^3}\right)\Bigg)\label{eq:local-decoding-construction}
\end{equation}
We claim that indexing the symbols of a communication over a synchronization channel with this string allows a repositioning algorithm which is both local and online. 
\begin{theorem}\label{thm:local-repositioning-for-long-distance}
For a communication over a synchronization channel that is indexed by string $R$ specified in \eqref{eq:local-decoding-construction}, there exists an online and local repositioning algorithm that guesses the position of each received symbol using only the symbol itself and $O_\eps(\log n)$ symbols preceding it in $O_\eps(\log^3 n)$ time. Also, among all symbols that are not deleted by the adversary, the position of no more than $\frac{n\delta}{1-\eps}$ will be
incorrectly guessed.
\end{theorem}
We remark that, similar to \cref{thm:online-reposition-misdecodings}, the number of incorrect guesses does not tend to zero by taking smaller values for $\eps$ and synchronization codes constructed using such an index string would achieve a rate of $1-3\delta-\eps^{O(1)}$.

\begin{proof}[Proof Sketch of \cref{thm:local-repositioning-for-long-distance}]
The details of the decoding algorithm and the proof of its properties are presented in \cite{haeupler2017synchronization3}. Here, we only give a high-level description of the algorithm and an informal justification of its correctness.

The analysis uses a notion called \emph{suffix error density}, that is the maximum density of errors that have occurred over all suffixes of the communication up to that point. The algorithm guarantees to make a correct guess as long as the suffix error density is less than $1-O(\eps)$.

The decoding algorithm begins with using the first and the third elements of \eqref{eq:local-decoding-construction} (that are the codeword of $\mathcal{C}$ and the circular counter) to find the vicinity of the location; more precisely, which codeword of $C$ the symbol belongs to or the quantity $\left\lfloor \frac{i}{N}\right\rfloor$. In order to do so, the algorithm looks at the $O\left(\frac{N}{\eps^2}\right)=O\left(\frac{\log n}{\eps^2}\right)$ last received symbols. If the value of the suffix error density is small, at least one of the previous $\frac{1}{\eps}$ codewords of $C$ appear inside this window with less than $1-O(\eps)$ synchronization errors. Let the value of the counter for the received symbol be $c$. With this observation, the repositioning algorithm runs the decoder of $C$ over the subsequences of symbols with counter values $c, c-1, c-2, \cdots, c-\frac{1}{\eps}$ in this window and comes up with a list of candidate vicinities for the position of the symbol. A second step uses the long-distance property of the string to choose one of these vicinities. Then, the exact position of the symbol within that vicinity is recovered using the repetitions of the small synchronization string that forms string $T$. (see \eqref{eq:local-decoding-construction})
\end{proof}

\subsection{Edit Distance Indexing and Near-linear Time Repositioning}\label{sec:edit-distance-approx-and-near-linear-repositioning}
As the final step in this section, we introduce a pseudo-random property for strings that, indexed to any given string, can facilitate edit distance computation. We then use such strings to enhance the construction of self-matching strings to obtain near-linear time repositioning algorithms.

\subsubsection{Edit Distance Approximation via Indexing}\label{sec:edit-ditance-indexes}
In this section, we introduce an indexing scheme which can be used to approximate edit distance in near-linear time if one of the strings is indexed by an \emph{edit-distance-approximating} string $I$. In particular, for every length $n$ and every $\eps >0$, one can, in near-linear time, construct a string $I \in \Sigma'^n$ with $|\Sigma'| = O_{\eps}(1)$, such that, indexing any string $S \in \Sigma^n$ with $I$ results in the string $S\times I \in \Sigma''^n$ where $\Sigma'' = \Sigma \times \Sigma'$ and there is an algorithm that approximates the edit distance between $S\times I$ and any other string within a $1+\eps$ factor and in $O(n\cdot
\poly(\log n))$ time.

The construction of the index string $I$ resembles the construction proposed for long-distance synchronization strings in \cref{sec:boosting-2}. Namely, the index string $I$ is simply constructed by writing, back-to-back, the codewords of a synchronization code that is list-decodable from high rates of synchronization errors, or more specifically, from any $1-\eps'$ fraction of insertions and $1-\eps'$ fraction of deletions for some $\eps' = \Theta(\eps)$. As we mentioned in \cref{sec:list-decodable-codes}, it is shown in \cite{haeupler2018synchronization4} that there exist families of codes that can be efficiently $L$-list-decoded from any $1-\eps'$ fraction of insertions and $1-\eps'$ fraction of deletions, achieve a rate of $\eps'/2$, and are defined over an alphabet of size $\exp(\eps'^{-3}\log \frac{1}{\eps'}) = O_{\eps'}(1)$. Here, $L$ is some sub-polynomial function of the block length of the code.

Note that the properties of the code directly implies that the index string $I$ is defined over an alphabet of size $O_\eps(1)$ and can be computed in near-linear time. Let us take a member of the above-mentioned family of codes like $C$ with block length $N$ and denote its block length and its decoding function respectively with $N$ and $\Decode_C(\cdot)$. Since the code has a positive rate, the length of the string formed by appending the codewords of $C$ together would be $n = \exp(N)$ and therefore, the construction time for index string $I$ is $\frac{n}{N}\cdot \poly(N) = O(n\cdot\polylog(n))$. 

We now describe the algorithm that approximates the edit distance of $S$ indexed with $I$ ($S\times I$) to any given string $S'$.
The algorithm starts by splitting the string $S'$ into blocks of length $N$ in the same spirit as $S\times I$. We denote the $i$th such block by $S'(i)$ and the $i$th block of $S\times I$ by $[S\times I](i)$. Note that the blocks of $S\times I$ are substrings of $S$ indexed by the codewords of an insertion-deletion code with high distance ($[S\times I](i)= S[N(i-1), Ni-1]\times C(i)$). 
Now, consider the set of insertions and deletions that correspond to the edit distance between $S\times I$ and $S'$ or n arbitrary one of them if there are more than one. 
One might expect that any block of $S$ that is not significantly altered by such insertions and deletions, (i) appears in a set of consecutive blocks in $S'$ and (ii) has a small edit distance to at least one of those blocks.

Following this intuition, our proposed algorithm works thusly: For any block of $S'$ like $S'(i)$, the algorithm uses the list-decoder of $\mathcal{C}$ to find all (up to $L$) blocks of $S$ that can be turned into $S'(i)$ by $N(1-\eps)$ deletions and $N(1-\eps)$ insertions only considering the index portion of the alphabet and ignoring the content portion of it. In other words, let $S'(i)=C'_i\times S'_0[N(i-1), Ni-1]$. We denote the set of such blocks by $\Decode_{\mathcal{C}}(C'_i)$.
Then, the algorithm constructs a bipartite graph $G$ with $|S|$ and $|S'|$ vertices on each side (representing symbols of $S$ and $S'$) as follows: a symbol in $S'(i)$ is connected to all identical symbols in the blocks that appear in $\Decode_{\mathcal{C}}(C'_i)$ or any block that is in their 
$w=O\left(\frac{1}{\eps}\right)$
 neighborhood, i.e., is up to 
$O\left(\frac{1}{\eps}\right)$
  blocks away from at least one of the members of $\Decode_{\mathcal{C}}(C'_i)$.

Note that any non-crossing matching in $G$ corresponds to some common subsequence between $S$ and $S'$ because $G$'s edges only connect identical symbols. In the next step, the algorithm finds the largest non-crossing matching in $G$, $\mathcal{M}_{ALG}$, and outputs the corresponding set of insertions and deletions as the output. Finally, an algorithm proposed by Hunt and Szymanski~\cite{hunt1977fast} is used to compute the largest non-crossing matching of $G$ with $n$ vertices and $r$ edges in $O\left((n+r)\log\log n\right)$. 
A formal description is available in~\cref{alg:EditDistanceApprox}.
As the number of edges of $G$ cannot exceed $\frac{n}{\log n}\cdot \log^2 n = n\log n$ and code $C$ is efficiently list-decodable, the run time for this algorithm is $O(n\cdot\polylog (n))$.

\begin{algorithm}
\caption{$(1+O(\eps'))$-Approximation for Edit Distance}\label{alg:EditDistanceApprox}
\begin{algorithmic}[1]
\REQUIRE{$S\times I, S', N, \Decode_{\mathcal{C}}(\cdot)$}

\STATE Make empty bipartite graph $G(|S|, |S'|)$
\STATE $w=\frac{1}{\eps'}$

\FOR{{\bf each} $S'(i)=C'_i\times S'_0[N(i-1), Ni-1]$}
\STATE{$List \leftarrow \Decode_{\mathcal{C}}(C'_i)$}
\FOR{{\bf each} $j \in List$}
\FOR{$k\in\left[j-w, j+w\right]$}
\STATE \label{step:adding-edges-to-G}Connect pairs of vertices in $G$ that correspond to identical symbols in $S(k)$ and $S'(i)$.
\ENDFOR
\ENDFOR
\ENDFOR

\STATE \label{step:largest-non-crossing}$\mathcal{M}_{ALG} \leftarrow$ Largest non-crossing matching in $G$ (\cite{hunt1977fast})

\ENSURE $\mathcal{M}_{ALG}$
\end{algorithmic}
\end{algorithm}

The detailed proof of the approximation guarantee is available in \cite{haeupler2019near}. We provide a general proof sketch here. 

Note that if graph $G$ from \cref{alg:EditDistanceApprox} contains the matching that corresponds to the $\LCS$ between $S\times I$ and $S'$, then the algorithm will find the longest common subsequence in Line \ref{step:largest-non-crossing} and compute the exact edit distance. To show that \cref{alg:EditDistanceApprox} finds a $1+O(\eps')$ approximation of the edit distance, \cite{haeupler2019near} associates any edge from the $\LCS$ missing in $G$ to $O(1/\eps')$ insertions or deletions from the optimal edit distance solution. 

Consider the matching that corresponds to the $\LCS$. If some block of $S'$ like $S'(i')$ is connected to more than $1+\frac{1}{\eps'}$ blocks in $S$, the unmatched vertices among those blocks account for $n\times \frac{1}{\eps'}$ deletions in the optimal edit distance solution. Therefore, even if none of the edges of $\LCS$ that have an endpoint in such blocks appear in $G$, the size of the edit distance would increase by a factor of $1+O(\eps')$. This is why the parameter $w$ is chosen as $\frac{1}{\eps'}$ in \cref{alg:EditDistanceApprox}.

Further, if some block of $S'$ is only connected to one block of $S$ and has no more than $\frac{N}{\eps}$ edges to it, $N - \frac{N}{\eps}$ of its symbols are insertions in the optimal edit distance solution. Therefore, the absence of its edges from $G$ in \cref{alg:EditDistanceApprox} may only increase the size of the edit distance solution by a factor of $1+O(\eps')$.

In \cite{haeupler2019near}, the authors show that all $\LCS$ edges that are absent from $G$ fall under these two categories and, therefore, the outcome of \cref{alg:EditDistanceApprox} is an $1+O(\eps')=1+\eps$ approximation.

\medskip
\subsubsection{Near-linear Time Repositioning}
Note that the repositioning algorithm for strings indexed with $\eps$-synchronization strings that was presented in \cref{alg:global-decoder} consists of multiple rounds of edit distance computation between the synchronization string used and a distorted version of it. 
To reduce the run time of the repositioning algorithm, one can use the edit-distance approximating indexes from \cref{sec:edit-ditance-indexes} and index $\eps$-synchronizations strings with them. Then, use edit distance approximations instead of exact computations in \cref{alg:global-decoder}. We formally summarize this in the following.

\begin{theorem}[Theorem 7.1 of \cite{haeupler2019near}]\label{lem:enhanced-sync-string-decoding}
Let $S$ be a string of length $n$ that consists of the symbol-wise concatenation of an $\eps_{s}$-synchronization string and an edit distance indexing sequence from \cref{sec:edit-ditance-indexes} with parameter $\eps_I$. Assume that a stream of messages indexed by $S$ goes through a channel that might impose up to $\delta \cdot n$ deletions and $\gamma\cdot n$ symbol insertions for some $0\le\delta<1$ and $0\le\gamma$. For any positive integer $K$, there exists a repositioning algorithm that runs in $O(Kn\cdot\polylog(n))$ time, guarantees up to 
$n\left(
\frac{1+\gamma}{K(1+\eps_I)}+ \frac{\eps_I(1+\gamma/2)}{1+\eps_I} +K\eps_s
\right)$
incorrect guesses and does not decode more than $K$ received symbols to any number in $[1, n]$.
\end{theorem}

\section{Further Applications of Synchronization Strings}\label{sec:other-applications}

\subsection{Codes for Block Transpositions and Replications}
We showed in \cref{sec:local-decoding} that using long-distance synchronization strings in the indexing-based synchronization code construction allows local repositioning, i.e., the decoder will be able to guess the original position of each symbol by only looking at a logarithmically long neighborhood of the received symbol. In this section, we show that this property enables the code to protect from block transposition and block duplication errors as well.

Block transposition errors allow for arbitrarily long substrings of the message to be moved to another position in the message string. Similarly, block duplication errors are ones that pick a substring of the message and copy it between two symbols of the communication.

We will present codes that can achieve a rate of $1-\delta-\eps$ and correct from some $O(\delta)$ fraction of synchronization errors, a $O(\delta/\log n)$ fraction of block errors, or a combination of them. A similar result for insertions, deletions, and block transpositions was shown by Schulman and Zuckerman \cite{schulman1999asymptotically} where they provided the first constant-distance and constant-rate synchronization code correcting from insertions, deletions, and block errors. They also show that the $O(\delta / \log n)$ resilience against block errors is optimal up to constants. 

\begin{theorem}
For any $0<r<1$ and sufficiently small $\eps$, there exists a code with rate $r$ that corrects $n\delta_{insdel}$ synchronization errors and $n\delta_{block}$ block transpositions or replications as long as $6\delta_{insdel} + (c\log n) \delta_{block} < 1-r-\eps$ for some $c=O(1)$. The code is over an alphabet of size $O_\eps(1)$ and has $O(n)$ encoding and $O(N\log^3 n)$ decoding complexities where $N$ is the length of the received message.
\end{theorem}
\begin{proof}[Proof Sketch]
Similar to \cref{sec:technical-warm-up}, this code is constructed by indexing near-MDS codes of Guruswami and Indyk~\cite{guruswami2005linear} with a pseudo-random string, particularly, long-distance synchronization strings. The decoding procedure also follows the same steps as \cref{sec:technical-warm-up}. Namely, the decoder uses the repositioning algorithm presented in~\cref{thm:local-repositioning-for-long-distance} to guess the actual position of the symbols and then runs the decoder of the Guruswami-Indyk code over the reconstructed string.

Note that the repositioning guarantee from~\cref{thm:local-repositioning-for-long-distance} implies that with the choice of some small $\eps$ parameter for the long-distance synchronization string, the repositioning algorithm correctly guesses the position of all but $O(n\delta_{insdel})$ symbols where $n$ is the length of the communication if only insertions and deletions are allowed.

Additionally, the local quality of the repositioning algorithm implies that any symbol at the receiver that does not have any synchronization errors or block error borders in its $O(\log n)$ neighborhood, is correctly repositioned by the local repositioning algorithm. Therefore, with $n\delta_{block}$ block errors, no more than $n\delta_{block}\log n$ repositioning guesses would be incorrect. This implies an $O(n/\log n)$ block error resilience.
Combining the two remarks above gives that the code can correct $n\delta_{insdel}$ synchronization errors and $n\delta_{block}$ block transpositions or replications as long as $6\delta_{insdel} + (c\log n) \delta_{block} < 1-r-\eps$ for some constant $c$.

The encoding and decoding complexities simply follow the properties of the Guruswami-Indyk codes, linear time constructions of long-distance synchronization strings from~\cref{thm:explicit-construction} and time complexity of the repositioning algorithm from~\cref{thm:local-repositioning-for-long-distance}.
\end{proof}

\subsection{Channel Simulation}\label{sec:channel-simulations}
The construction of codes based on indexing presented in this paper suggests that indexing with pseudo-random strings can reduce synchronization errors to more benign Hamming-type errors (substitutions and erasures). In this section, we present results from~\cite{haeupler2017synchronization2,haeupler2017synchronization3} which shows that this is indeed true. 

More precisely, having a channel afflicted by synchronization errors, one can put two simulation agents on the two ends of the channel who can simulate a channel with Hamming-type errors over the given channel. In other words, the sender/receiver sends/receives symbols to/from their corresponding agent and the simulation guarantees that the channel would seem like a channel with Hamming-type errors to the parties.

Note that the indexing scheme from \cref{sec:technical-warm-up} almost achieves this goal by reducing synchronization errors to half-errors through indexing. However, this procedure requires all symbols to be communicated before the repositioning procedure of \cref{alg:global-decoder} can start running and, therefore, introduces a delay. A true channel simulation would not add such delay. More precisely, a round of error-free communication in a simulated channel is one that communicates the $i$th symbol sent by the sender as the $i$th symbol to the receiver once it arrives at the other side and prior to the $i+1$st symbol being sent by the sender.

This subtle requirement can be satisfied through using synchronization strings as the indexing sequence and utilizing the online repositioning algorithm introduced in \cref{sec:sync-strings-and-online-decoding}. 

Before presenting the channel simulations, we remark an interesting negative result of \cite{haeupler2017synchronization2} stating that, as opposed to codes, when it comes to channel simulations, no channel simulator can reduce $\delta$ fraction of synchronization errors to $\delta+\eps$ half-errors for arbitrarily small $\eps$.

\begin{theorem}\label{thm:OnewaySimulLowerBound}
Assume that $n$ uses of a synchronization channel over an arbitrarily large alphabet $\Sigma$ with a $\delta$ fraction of insertions and deletions are given. There is no deterministic simulation of a half-error channel over any alphabet $\Sigma_{sim}$ where the simulated channel guarantees more than $n\left(1-4\delta/3\right)$ uncorrupted transmitted symbols. If the simulation is randomized, the expected number of uncorrupted transmitted symbols is at most $n(1-7\delta/6)$.
\end{theorem}

We now present the channel simulations that can be achieved via indexing with synchronization strings. Simulations are presented for channels with large constant alphabets, binary alphabets, one-way communication or interactive communication.

\begin{theorem}\label{thm:nearLinearChannelSimulationsForInsdel}[Channel Simulations]
\-\begin{enumerate}
\item[(a)] Suppose that $n$ rounds of a one-way/interactive insertion-deletion channel over an alphabet $\Sigma$ with a $\delta$ fraction of insertions and deletions are given. Using a long-distance $\eps$-synchronization string over alphabet $\Sigma_{syn}$, it is possible to simulate $n\left(1-O_\eps(\delta)\right)$ rounds of a one-way/interactive substitution channel over $\Sigma_{sim}$ with at most $O_\eps\left(n\delta\right)$ symbols corrupted so long as $|\Sigma_{sim}| \times |\Sigma_{syn}| \le |\Sigma|$. 
\item[(b)] Suppose that $n$ rounds of a binary one-way/interactive insertion-deletion channel with a $\delta$ fraction of insertions and deletions are given. It is possible to simulate 
$n(1-\Theta( \sqrt{\delta\log(1/\delta)}))$
 rounds of a binary one-way/interactive substitution channel 
 with $\Theta(\sqrt{\delta\log(1/\delta)})$ fraction of substitution errors between two parties over the given channel.
\end{enumerate}
All of the simulations mentioned above take $O(1)$ time per symbol for the sending/starting party of one-way/interactive communications. Further, on the other side, the simulation spends $O(\log^3 n)$ time upon arrival of each symbol and only looks up $O(\log n)$ recently received symbols. Overall, these simulations take a $O(n\log^3 n)$ time and $O(\log n)$ space to run. These simulations can be performed even if parties are not aware of the communication length.
\end{theorem}
\begin{proof}[Proof Sketch]
We highlight the main ideas behind each of these simulations in the following.
\begin{enumerate}
    \item[(a)] In this simulation, the simulating agents do the indexing as done in the case of coding. Meaning that on the sender side, the simulation simply indexes the messages of the sender with symbols of a long-distance synchronization string and on the receiving end, the receiver-side simulating agent runs the online repositioning algorithm from \cref{sec:local-decoding} to identify the position of the symbols it receives and relays them to the receiver.
    
    Note that the online repositioning algorithm for long-distance synchronization strings allows the simulator on the receiving end to guess the positions of the received symbols as they arrive. However, we stress that the simulated channel has to behave as an actual channel and therefore cannot reveal the symbols to the receiver out of order. For instance, if the repositioning algorithm incorrectly identifies the first symbol as the tenth symbol and reveals it as the tenth symbol to the receiver, it cannot reveal the second to ninth symbols to the receiver afterwards even if the repositioning for those symbols is done correctly. 
    
    To ensure in-order revealing of symbols, the simulation uses a lazy revealing strategy to avoid over-reacting to incorrect guesses by the repositioning algorithm. More precisely, if the guessed position of a symbol is far beyond where the communication length is at that moment, the receiver-side simulator moves the communication forward by outputing two dummy symbols to the reciever. For the analysis of this strategy, we refer the reader to \cite{haeupler2017synchronization2}.
    
    \item[(b)] The simulation for channels with a binary alphabet is very similar except that the indexing is not possible due to the size of the alphabet. To overcome this, the simulation splits the communication into several blocks. In each block, the sender-side simulator first sends a fixed header of size $O(\log \frac{1}{\delta})$ (indicating the start of a new block), then sends a binary encoding of a symbol of the long-distance synchronization string, and then ends the block by relaying the messages of the sender for $r=\sqrt{\frac{\log 1/\delta}{\delta}}$ rounds.
\end{enumerate}
Time and space guarantees of these simulations are inferred from highly-explicit constructions of infinite long-distance synchronization strings (\cref{thm:infinite-construction}) and their local repositioning algorithms (\cref{thm:local-repositioning-for-long-distance}).
Similar simulations can be performed in interactive communication channels by taking the steps mentioned above in one direction of the communication.
\end{proof}

\subsection{Interactive Communication for Synchronization Errors}
The channel simulations via indexing presented in \cref{sec:channel-simulations} can be used to obtain interactive coding schemes for synchronization errors. Interactive communication between two parties is one in which any round of the communication consists of a message transmission from one party to the other one. Each party is assumed to hold a private information denoted by $X$ and $Y$ and the goal is for both parties to compute some function $f(X, Y)$. Any strategy for computing $f(X, Y)$ is called a protocol. 

A coding scheme for interactive communication is one that takes any protocol that computes some function $f$ in noiseless communication and converts it into a protocol that computes $f$ over a noisy channel. 
The rate of an interactive coding scheme is defined as the minimal ratio of the length of the protocol in the absence of noise over the length of the protocol in the presence of noise over all functions $f$.

The channel model used in the results of this section is the commonly used model of Braverman \emph{et al}.~\cite{braverman2015coding} that considers an alternating protocol, i.e., protocols in which parties take alternating turns in sending and receiving symbols.

Using channel simulations, \cite{haeupler2017synchronization2,haeupler2017synchronization3} provide coding schemes for interactive communication over channels suffering from synchronization errors by simply simulating a half-error channel over the given synchronization channel and applying interactive protocols for channels with symbol substitution errors over the simulated channel.
Using simulations for channels with large alphabets along with the interactive protocol of Haeupler and Ghaffari~\cite{ghaffari2014optimal}, \cite{haeupler2017synchronization3} gives the following. 

\begin{theorem}
For a sufficiently small $\delta$ and $n$-round alternating protocol $\Pi$, there is a randomized coding scheme simulating $\Pi$ in the presence of $\delta$ fraction of synchronization errors with constant rate (i.e., in $O(n)$ rounds) and in near-linear time. This coding scheme works with probability $1-2^{\Theta(n)}$. 
\end{theorem}

Similarly, using binary alphabet simulations and the interactive protocol of Haeupler~\cite{haeupler2014interactive:FOCS}, \cite{haeupler2017synchronization2} gives the following.
\begin{theorem}\label{thm:InterFullyAdv}
For sufficiently small $\delta$, there is an efficient interactive coding scheme for fully adversarial binary synchronization channels which is robust against $\delta$ fraction of edit-corruptions, achieves a communication rate of 
$1 - \Theta({\sqrt{\delta\log(1/\delta)}})$, and works with probability $1 - 2^{-\Theta(n\delta)}$.
\end{theorem}

\subsection{Binary Synchronization Codes}
A similar approach is taken to design binary synchronization error-correcting codes in \cite{haeupler2017synchronization2} by simulating a half-error channel over the given synchornization channel and then using a binary error-correcting code on top of it.
\begin{theorem}\label{thm:binary-insdel-sync-str}
For any sufficiently small $\delta$, there is a binary synchronization code with rate $1-\Theta\left(\sqrt{\delta\log\frac{1}{\delta}}\right)$ which is decodable from $\delta$ fraction of insertions and deletions. 
\end{theorem}
It is shown in \cite{levenshtein1966binary} that the optimal rate for binary synchronization codes with distance $\delta$ is $1-O\left(\delta\log\frac{1}{\delta}\right)$. Recent works by Cheng \emph{et al}.~\cite{cheng2018deterministic} and Haeupler~\cite{haeupler2018optimal} have simultaneously improved over the codes from \cref{thm:binary-insdel-sync-str} by introducing efficient binary codes with rate $1-O(\delta\log^2\frac{1}{\delta})$ via providing deterministic document exchange protocols.

\subsection{Document Exchange}
Document exchange is a problem in which a server and a client hold two versions of the same string, say $F$ and $F'$ respectively, where $F'$ is an outdated version that is different from $F$ by up to $k$ insertions or deletions. The goal is for the server to compute a small summary and send it to the client so the client can update its string to $F$.

There is a close connection between deterministic document exchange protocols and systematic synchronization codes. Having a systematic synchronization code, one can construct a document exchange protocol with using the non-systematic part of the code as the summary. On the other hand, having a document exchange protocol, one can construct  a systematic code by taking the summary of the document exchange protocol, encoding it using a synchronization code, and using the encoded summary as the non-systematic part of the code.

For document exchange with $k$ errors, $\Omega(k \log \frac{n}{k})$ bits of information is necessary as the summary. Orlitsky~\cite{orlitsky1991interactive} showed in 1991 that protocols with this amount of redundancy exist, however, fell short of providing efficient ones. 
In 2005, Irmak, Mihaylov and Suel~\cite{irmak2005improved} provided an efficient document exchange protocol with $O(k \log \frac{n}{k} \log n)$ redundancy. Since then, there have been several works on randomized document exchange protocols (mostly for $k$ sublinear in $n$) by \cite{dodis2008fuzzy,jowhari2012efficient,chakraborty2016streaming,belazzougui2016edit}.
Recent works of Cheng \emph{et al}.~\cite{cheng2018deterministic} and Haeupler~\cite{haeupler2018optimal} provide deterministic document exchange protocols with redundancy $O(k\log^2\frac{n}{k})$.

Haeupler~\cite{haeupler2018optimal} first provides a randomized document exchange protocol with redundancy $O(\delta\log\frac{1}{\delta})$ through a modification and careful analysis of the protocol of Irmak \emph{et al}.~\cite{irmak2005improved}. Then, it derandomizes the protocol using a derandomizarion technique reminiscent of one used in \cite{haeupler2017synchronization3} of synchronization strings. The techniques of Cheng \emph{et al}.~\cite{cheng2018deterministic} also make use of notions called $\eps$-self-matching hash functions and $\eps$-synchronization hash functions which are sequences of hash functions whose outputs satisfy properties resembling the corresponding string properties introduced in~\cite{haeupler2017synchronization}.

\subsection{Linear Insertion-Deletion Codes}
A recent work of Cheng \emph{et al.}~\cite{cheng2020efficient} studies qualities like linearity and affinity in the context of synchronization coding. They propose an efficient synchronization string-based transformation that can convert any asymptotically good linear error-correcting code into an asymptotically good insertion-deletion code. Using this transformation along with well-known linear error-correcting codes, such as Hamming codes, results in explicit constructions for linear insertion-deletion codes.

The indexing scheme introduced in \cref{sec:main-results} inherently leads to codes that are non-linear as it specifies a fixed value to a portion of each symbol. To circumvent this, \cite{cheng2020efficient} uses pseudo-random strings to design linear insertion-deletion codes in the following manner: To encode a message $x \in \mathbb{F}_q^m$, they first encode it using a linear error-correcting code for Hamming-type errors $C:   \mathbb{F}_q^m \rightarrow  \mathbb{F}_q^{n}$ to become $y = C(x)$. They then take care of the synchronization issues by inserting several sequences of $0$ symbols into the message and generating the final codeword as follows:
$$z_? = (0^{S_1}, y_1, 0^{S_2}, y_2, \cdots, 0^{S_n}, y_n).$$
The string  $S = (S_1, S_2, \cdots, S_n)$ is a pseudo-random string having synchronization properties similar to the ones studied in this survey.

More precisely, \cite{cheng2020efficient} defines a \emph{$\Lambda$-synchronization separator sequence} as a sequence $S$ for which any $z=(0^{S_1}, ?, 0^{S_2}, ?, \cdots, 0^{S_n}, ?)$ does not have self-matchings with more than $\Lambda$ \emph{undesirable} matches. An undesirable match is one between two `?'s like the $i$th and the $j$th `?' where $i\neq j$ and $p_i-p_{i'} =  p_j - p_{j'}$ where $(i', j')$ is the immediate previous match to $(i, j)$ in the matching and $p_i$ denote the position of the $i$th `?' in $z_?$. 

\cite{cheng2020efficient} provides explicit constructions for synchronization separator sequences and shows that, if used in the above construction, they enable the decoder to reconstruct the codeword $y$ up to a number of Hamming-type errors that is within a constant factor of the number of insertions and deletions applied; Hence, proving that this conversion preserves both linearity and the asymptotic goodness of the code.

\subsection{Coded Trace Reconstruction}
A recent work by Brakensiek \emph{et al.}~\cite{brakensiek2019coded} provides novel results for the coded trace reconstruction problem. Coded trace reconstruction asks for codes that satisfy the following: Assuming that a sender chooses a codeword of the code and sends multiple copies of it over independent binary deletion channels (called \emph{traces}), the receiver wants to be able to recover the original codeword with high probability. 
Using synchronization strings, \cite{brakensiek2019coded} provides a high-rate coded trace reconstruction scheme that is efficiently decodable from a constant number of traces.

\subsection{Coding for Binary Deletion Channels and Poisson Repeat Channels}
Con and Shpilka~\cite{con2020explicit} use the synchronization string-based Singleton-bound-approaching synchronization codes from \cref{sec:intro-main-result-1} to provide an efficient and explicit code for binary deletion channels that improve over the state-of-the-art in terms of error resilience. They, additionally, show that their code also works for the Poisson repeat channel where each bit appears on the receiver's side a number of times which follows some Poisson distribution.

\bibliographystyle{plain}
\bibliography{bibliography}

\begin{thebibliography}{10}

\bibitem{abdel2011helberg}
Khaled A~S Abdel-Ghaffar, Filip Paluncic, Hendrik~C Ferreira, and Willem~A
  Clarke.
\newblock On {H}elberg's generalization of the {L}evenshtein code for multiple
  deletion/insertion error correction.
\newblock {\em IEEE Transactions on Information Theory}, 58(3):1804--1808,
  2011.

\bibitem{arikan2008channel}
Erdal Arikan.
\newblock Channel polarization: a method for constructing capacity-achieving
  codes for symmetric binary-input memoryless channels.
\newblock {\em IEEE Transactions on Information Theory}, 55(7):3051--3073,
  2009.

\bibitem{belazzougui2016edit}
Djamal Belazzougui and Qin Zhang.
\newblock Edit distance: sketching, streaming, and document exchange.
\newblock In {\em Proceedings of the IEEE Symposium on Foundations of Computer
  Science (FOCS)}, pages 51--60, 2016.

\bibitem{blasiok2018polar}
Jaroslaw Blasiok, Venkatesan Guruswami, Preetum Nakkiran, Atri Rudra, and Madhu
  Sudan.
\newblock General strong polarization.
\newblock In {\em Proceedings of the ACM Symposium on Theory of Computing
  (STOC)}, pages 485--492, 2018.

\bibitem{blawat2016forward}
Meinolf Blawat, Klaus Gaedke, Ingo Huetter, Xiao-Ming Chen, Brian Turczyk,
  Samuel Inverso, et~al.
\newblock Forward error correction for {DNA} data storage.
\newblock {\em Procedia Computer Science}, 80:1011--1022, 2016.

\bibitem{bornholt2016dna}
James Bornholt, Randolph Lopez, Douglas~M Carmean, Luis Ceze, Georg Seelig, and
  Karin Strauss.
\newblock A {DNA}-based archival storage system.
\newblock {\em ACM SIGARCH Computer Architecture News}, 44(2):637--649, 2016.

\bibitem{brakensiek2017efficient}
Joshua Brakensiek, Venkatesan Guruswami, and Samuel Zbarsky.
\newblock Efficient low-redundancy codes for correcting multiple deletions.
\newblock {\em IEEE Transactions on Information Theory}, 64(5):3403--3410,
  2017.

\bibitem{brakensiek2019coded}
Joshua Brakensiek, Ray Li, and Bruce Spang.
\newblock Coded trace reconstruction in a constant number of traces.
\newblock In {\em Proceedings of the IEEE Symposium on Foundations of Computer
  Science (FOCS)}, 2020.

\bibitem{braverman2015coding}
Mark Braverman, Ran Gelles, Jieming Mao, and Rafail Ostrovsky.
\newblock Coding for interactive communication correcting insertions and
  deletions.
\newblock In {\em Proceedings of the International Conference on Automata,
  Languages, and Programming (ICALP)}, 2016.

\bibitem{bukh2017improved}
Boris Bukh, Venkatesan Guruswami, and Johan H{\aa}stad.
\newblock An improved bound on the fraction of correctable deletions.
\newblock {\em IEEE Transactions on Information Theory}, 63(1):93--103, 2017.

\bibitem{bukh2014longest}
Boris Bukh and Jie Ma.
\newblock Longest common subsequences in sets of words.
\newblock {\em SIAM Journal on Discrete Mathematics}, 28(4):2042--2049, 2014.

\bibitem{chakraborty2016streaming}
Diptarka Chakraborty, Elazar Goldenberg, and Michal Kouck{\`y}.
\newblock Streaming algorithms for embedding and computing edit distance in the
  low distance regime.
\newblock In {\em Proceedings of the ACM Symposium on Theory of Computing
  (STOC)}, pages 712--725, 2016.

\bibitem{chandrasekaran2013deterministic}
Karthekeyan Chandrasekaran, Navin Goyal, and Bernhard Haeupler.
\newblock Deterministic algorithms for the lov{\'a}sz local lemma.
\newblock {\em SIAM Journal on Computing}, 42(6):2132--2155, 2013.

\bibitem{cheng2020efficient}
Kuan Cheng, Venkatesan Guruswami, Bernhard Haeupler, and Xin Li.
\newblock Efficient linear and affine codes for correcting
  insertions/deletions.
\newblock In {\em Proceedings of the ACM-SIAM Symposium on Discrete Algorithms
  (SODA)}, 2020.

\bibitem{cheng2018synchronization}
Kuan Cheng, Bernhard Haeupler, Xin Li, Amirbehshad Shahrasbi, and Ke~Wu.
\newblock Synchronization strings: highly efficient deterministic constructions
  over small alphabets.
\newblock In {\em Proceedings of the ACM-SIAM Symposium on Discrete Algorithms
  (SODA)}, pages 2185--2204, 2019.

\bibitem{cheng2018deterministic}
Kuan Cheng, Zhengzhong Jin, Xin Li, and Ke~Wu.
\newblock Deterministic document exchange protocols, and almost optimal binary
  codes for edit errors.
\newblock In {\em Proceedings of the IEEE Symposium on Foundations of Computer
  Science (FOCS)}, 2018.

\bibitem{cheraghchi2019overview}
Mahdi Cheraghchi and Jo{\~a}o Ribeiro.
\newblock An overview of capacity results for synchronization channels.
\newblock {\em IEEE Transactions on Information Theory}, 2020.

\bibitem{church2012next}
George~M Church, Yuan Gao, and Sriram Kosuri.
\newblock Next-generation digital information storage in {DNA}.
\newblock {\em Science}, 337(6102):1628--1628, 2012.

\bibitem{con2020explicit}
Roni Con and Amir Shpilka.
\newblock Explicit and efficient constructions of coding schemes for the binary
  deletion channel.
\newblock In {\em Proceedings of the IEEE International Symposium on
  Information Theory (ISIT)}, pages 84--89, 2020.

\bibitem{dodis2008fuzzy}
Yevgeniy Dodis, Rafail Ostrovsky, Leonid Reyzin, and Adam Smith.
\newblock Fuzzy extractors: How to generate strong keys from biometrics and
  other noisy data.
\newblock {\em SIAM Journal on Computing}, 38(1):97--139, 2008.

\bibitem{ferreira1997insertion}
WC~Ferreira, Willem~A Clarke, Albertus S.~J. Helberg, Khaled A.~S.
  Abdel-Ghaffar, and AJ~Han Vinck.
\newblock Insertion/deletion correction with spectral nulls.
\newblock {\em IEEE Transactions on Information Theory}, 43(2):722--732, 1997.

\bibitem{gabrys2018codes}
Ryan Gabrys and Frederic Sala.
\newblock Codes correcting two deletions.
\newblock {\em IEEE Transactions on Information Theory}, 65(2):965--974, 2018.

\bibitem{ghaffari2014optimal}
Mohsen Ghaffari and Bernhard Haeupler.
\newblock Optimal error rates for interactive coding {II}: Efficiency and list
  decoding.
\newblock In {\em Proceedings of the IEEE Symposium on Foundations of Computer
  Science (FOCS)}, pages 394--403, 2014.

\bibitem{gilbert1960synchronization}
E~Gilbert.
\newblock Synchronization of binary messages.
\newblock {\em IRE Transactions on Information Theory}, 6(4):470--477, 1960.

\bibitem{goldman2013towards}
Nick Goldman, Paul Bertone, Siyuan Chen, Christophe Dessimoz, Emily~M LeProust,
  Botond Sipos, and Ewan Birney.
\newblock Towards practical, high-capacity, low-maintenance information storage
  in synthesized {DNA}.
\newblock {\em Nature}, 494(7435):77, 2013.

\bibitem{golomb1963synchronization}
SW~Golomb, J~Davey, I~Reed, H~Van~Trees, and J~Stiffler.
\newblock Synchronization.
\newblock {\em IEEE Transactions on Communications Systems}, 11(4):481--491,
  1963.

\bibitem{guibas1978maximal}
Leonidas~J Guibas and Andrew~M Odlyzko.
\newblock Maximal prefix-synchronized codes.
\newblock {\em SIAM Journal on Applied Mathematics}, 35(2):401--418, 1978.

\bibitem{guruswami2019optimally}
Venkatesan Guruswami, Bernhard Haeupler, and Amirbehshad Shahrasbi.
\newblock Optimally resilient codes for list-decoding from insertions and
  deletions.
\newblock In {\em Proceedings of the ACM Symposium on Theory of Computing
  (STOC)}, pages 524--537, 2020.

\bibitem{guruswami2005linear}
Venkatesan Guruswami and Piotr Indyk.
\newblock Linear-time encodable/decodable codes with near-optimal rate.
\newblock {\em IEEE Transactions on Information Theory}, 51(10):3393--3400,
  2005.

\bibitem{guruswami2016efficiently}
Venkatesan Guruswami and Ray Li.
\newblock Efficiently decodable insertion/deletion codes for high-noise and
  high-rate regimes.
\newblock In {\em Proceedings of the IEEE International Symposium on
  Information Theory (ISIT)}, pages 620--624, 2016.

\bibitem{guruswami2017deletion}
Venkatesan Guruswami and Carol Wang.
\newblock Deletion codes in the high-noise and high-rate regimes.
\newblock {\em IEEE Transactions on Information Theory}, 63(4):1961--1970,
  2017.

\bibitem{haeupler2014interactive:FOCS}
Bernhard Haeupler.
\newblock Interactive channel capacity revisited.
\newblock In {\em Proceedings of the IEEE Symposium on Foundations of Computer
  Science (FOCS)}, pages 226--235, 2014.

\bibitem{haeupler2018optimal}
Bernhard Haeupler.
\newblock Optimal document exchange and new codes for insertions and deletions.
\newblock In {\em Proceedings of the IEEE Symposium on Foundations of Computer
  Science (FOCS)}, pages 334--347, 2019.

\bibitem{haeupler2019near}
Bernhard Haeupler, Aviad Rubinstein, and Amirbehshad Shahrasbi.
\newblock Near-linear time insertion-deletion codes and
  (1+$\varepsilon$)-approximating edit distance via indexing.
\newblock In {\em Proceedings of the ACM Symposium on Theory of Computing
  (STOC)}, pages 697--708, 2019.

\bibitem{haeupler2017synchronization}
Bernhard Haeupler and Amirbehshad Shahrasbi.
\newblock Synchronization strings: Codes for insertions and deletions
  approaching the {Singleton} bound.
\newblock In {\em Proceedings of the ACM Symposium on Theory of Computing
  (STOC)}, pages 33--46, 2017.

\bibitem{haeupler2017synchronization3}
Bernhard Haeupler and Amirbehshad Shahrasbi.
\newblock Synchronization strings: Explicit constructions, local decoding, and
  applications.
\newblock In {\em Proceedings of the ACM Symposium on Theory of Computing
  (STOC)}, pages 841--854, 2018.

\bibitem{haeupler-list-dec-capacity2020}
Bernhard Haeupler and Amirbehshad Shahrasbi.
\newblock Rate-distance tradeoffs for list-decodable insertion-deletion codes.
\newblock {\em arXiv preprint arXiv:2009.13307}, 2020.

\bibitem{haeupler2018synchronization4}
Bernhard Haeupler, Amirbehshad Shahrasbi, and Madhu Sudan.
\newblock Synchronization strings: List decoding for insertions and deletions.
\newblock In {\em Proceedings of the International Conference on Automata,
  Languages, and Programming (ICALP)}, 2018.

\bibitem{haeupler2017synchronization2}
Bernhard Haeupler, Amirbehshad Shahrasbi, and Ellen Vitercik.
\newblock Synchronization strings: Channel simulations and interactive coding
  for insertions and deletions.
\newblock In {\em Proceedings of the International Conference on Automata,
  Languages, and Programming (ICALP)}, pages 75:1--75:14, 2018.

\bibitem{hayashi2018list}
Tomohiro Hayashi and Kenji Yasunaga.
\newblock On the list decodability of insertions and deletions.
\newblock In {\em Proceedings of the IEEE International Symposium on
  Information Theory (ISIT)}, pages 86--90, 2018.

\bibitem{helberg2002multiple}
Albertus S~J Helberg and Hendrik~C Ferreira.
\newblock On multiple insertion/deletion correcting codes.
\newblock {\em IEEE Transactions on Information Theory}, 48(1):305--308, 2002.

\bibitem{hemenway2017local}
Brett Hemenway, Noga Ron-Zewi, and Mary Wootters.
\newblock Local list recovery of high-rate tensor codes and applications.
\newblock {\em Proceedings of the IEEE Symposium on Foundations of Computer
  Science (FOCS)}, 2017.

\bibitem{hunt1977fast}
James~W Hunt and Thomas~G Szymanski.
\newblock A fast algorithm for computing longest common subsequences.
\newblock {\em Communications of the ACM}, 20(5):350--353, 1977.

\bibitem{irmak2005improved}
Utku Irmak, Svilen Mihaylov, and Torsten Suel.
\newblock Improved single-round protocols for remote file synchronization.
\newblock In {\em Proceedings of the Annual Joint Conference of the IEEE
  Computer and Communications Societies}, volume~3, pages 1665--1676, 2005.

\bibitem{jowhari2012efficient}
Hossein Jowhari.
\newblock Efficient communication protocols for deciding edit distance.
\newblock In {\em European Symposium on Algorithms}, pages 648--658. Springer,
  2012.

\bibitem{kautz1965fibonacci}
W~Kautz.
\newblock Fibonacci codes for synchronization control.
\newblock {\em IEEE Transactions on Information Theory}, 11(2):284--292, 1965.

\bibitem{kopparty2019list}
Swastik Kopparty, Nicolas Resch, Noga Ron-Zewi, Shubhangi Saraf, and Shashwat
  Silas.
\newblock On list recovery of high-rate tensor codes.
\newblock {\em IEEE Transactions on Information Theory}, 2020.

\bibitem{levenshtein1966binary}
Vladimir Levenshtein.
\newblock Binary codes capable of correcting deletions, insertions, and
  reversals.
\newblock {\em Doklady Akademii Nauk SSSR}, 163(4):845--848, 1965.
\newblock {E}nglish translation in \emph{Soviet Physics Doklady},
  10(8):707--710, 1966.

\bibitem{liu2019list}
Shu Liu, Ivan Tjuawinata, and Chaoping Xing.
\newblock List decoding of insertion and deletion codes.
\newblock {\em arXiv preprint arXiv:1906.09705}, 2019.

\bibitem{mercier2010survey}
Hugues Mercier, Vijay~K Bhargava, and Vahid Tarokh.
\newblock A survey of error-correcting codes for channels with symbol
  synchronization errors.
\newblock {\em IEEE Communications Surveys \& Tutorials}, 12(1):87--96, 2010.

\bibitem{mitzenmacher2009survey}
Michael Mitzenmacher.
\newblock A survey of results for deletion channels and related synchronization
  channels.
\newblock {\em Probability Surveys}, 6:1--33, 2009.

\bibitem{morita1997prefix}
H~Morita, AJ~Van~Wijngaarden, and AJ~Han Vinck.
\newblock Prefix synchronized codes capable of correcting single
  insertion/deletion errors.
\newblock In {\em Proceedings of the IEEE International Symposium on
  Information Theory (ISIT)}, page 409, 1997.

\bibitem{morita1996construction}
Hiroyoshi Morita, Adriaan~J van Wijngaarden, and AJ~Han Vinck.
\newblock On the construction of maximal prefix-synchronized codes.
\newblock {\em IEEE Transactions on Information Theory}, 42(6):2158--2166,
  1996.

\bibitem{organick2017scaling}
Lee Organick, Siena~Dumas Ang, Yuan-Jyue Chen, Randolph Lopez, Sergey Yekhanin,
  Konstantin Makarychev, et~al.
\newblock Scaling up {DNA} data storage and random access retrieval.
\newblock {\em BioRxiv}, 2017.

\bibitem{orlitsky1991interactive}
Alon Orlitsky.
\newblock Interactive communication: Balanced distributions, correlated files,
  and average-case complexity.
\newblock In {\em Proceedings of the IEEE Symposium on Foundations of Computer
  Science (FOCS)}, pages 228--238, 1991.

\bibitem{Rubinstein18-blog}
Aviad Rubinstein.
\newblock Approximating edit distance.
\newblock
  \url{https://theorydish.blog/2018/07/20/approximating-edit-distance/}, 2018.

\bibitem{schulman1999asymptotically}
Leonard~J. Schulman and David Zuckerman.
\newblock Asymptotically good codes correcting insertions, deletions, and
  transpositions.
\newblock {\em IEEE Transactions on Information Theory}, 45(7):2552--2557,
  1999.

\bibitem{sellers1962bit}
F~Sellers.
\newblock Bit loss and gain correction code.
\newblock {\em IRE Transactions on Information Theory}, 8(1):35--38, 1962.

\bibitem{sloane2002single}
Neil~JA Sloane.
\newblock On single-deletion-correcting codes.
\newblock {\em Codes and Designs}, 10:273--291, 2002.

\bibitem{tenengolts1984nonbinary}
Grigory Tenengolts.
\newblock Nonbinary codes, correcting single deletion or insertion (corresp.).
\newblock {\em IEEE Transactions on Information Theory}, 30(5):766--769, 1984.

\bibitem{van1995extended}
AJ~Van~Wijngaarden and B~Morita.
\newblock Extended prefix synchronization codes.
\newblock In {\em Proceedings of the IEEE International Symposium on
  Information Theory (ISIT)}, page 465, 1995.

\bibitem{wachter2017list}
Antonia Wachter-Zeh.
\newblock List decoding of insertions and deletions.
\newblock {\em IEEE Transactions on Information Theory}, 64(9):6297--6304,
  2018.

\bibitem{yazdi2015dna}
SM~Hossein~Tabatabaei Yazdi, Han~Mao Kiah, Eva Garcia-Ruiz, Jian Ma, Huimin
  Zhao, and Olgica Milenkovic.
\newblock {DNA}-based storage: Trends and methods.
\newblock {\em IEEE Transactions on Molecular, Biological and Multi-Scale
  Communications}, 1(3):230--248, 2015.

\end{thebibliography}

\end{document}